\documentclass{article} 

\usepackage{amsmath}
\usepackage{amssymb}
\usepackage{amstext}
\usepackage{xspace}
\usepackage{url}
\usepackage{xy}
\xyoption{all}
\usepackage{proof}

\newdimen\proofrulebreadth \proofrulebreadth=.05em
\newdimen\proofdotseparation \proofdotseparation=1.25ex
\newdimen\proofrulebaseline \proofrulebaseline=2ex
\newcount\proofdotnumber \proofdotnumber=3
\let\then\relax
\def\hfi{\hskip0pt plus.0001fil}
\mathchardef\squigto="3A3B
\newif\ifinsideprooftree\insideprooftreefalse
\newif\ifonleftofproofrule\onleftofproofrulefalse
\newif\ifproofdots\proofdotsfalse
\newif\ifdoubleproof\doubleprooffalse
\let\wereinproofbit\relax
\newdimen\shortenproofleft
\newdimen\shortenproofright
\newdimen\proofbelowshift
\newbox\proofabove
\newbox\proofbelow
\newbox\proofrulename
\def\shiftproofbelow{\let\next\relax\afterassignment\setshiftproofbelow\dimen0 }
\def\shiftproofbelowneg{\def\next{\multiply\dimen0 by-1 }%
\afterassignment\setshiftproofbelow\dimen0 }
\def\setshiftproofbelow{\next\proofbelowshift=\dimen0 }
\def\setproofrulebreadth{\proofrulebreadth}

\def\prooftree{
\ifnum  \lastpenalty=1
\then   \unpenalty
\else   \onleftofproofrulefalse
\fi
\ifonleftofproofrule
\else   \ifinsideprooftree
        \then   \hskip.5em plus1fil
        \fi
\fi
\bgroup
\setbox\proofbelow=\hbox{}\setbox\proofrulename=\hbox{}%
\let\justifies\proofover\let\leadsto\proofoverdots\let\Justifies\proofoverdbl
\let\using\proofusing\let\[\prooftree
\ifinsideprooftree\let\]\endprooftree\fi
\proofdotsfalse\doubleprooffalse
\let\thickness\setproofrulebreadth
\let\shiftright\shiftproofbelow \let\shift\shiftproofbelow
\let\shiftleft\shiftproofbelowneg
\let\ifwasinsideprooftree\ifinsideprooftree
\insideprooftreetrue
\setbox\proofabove=\hbox\bgroup$\displaystyle 
\let\wereinproofbit\prooftree
\shortenproofleft=0pt \shortenproofright=0pt \proofbelowshift=0pt
\onleftofproofruletrue\penalty1
}

\def\eproofbit{
\ifx    \wereinproofbit\prooftree
\then   \ifcase \lastpenalty
        \then   \shortenproofright=0pt  
        \or     \unpenalty\hfil         
        \or     \unpenalty\unskip       
        \else   \shortenproofright=0pt  
        \fi
\fi
\global\dimen0=\shortenproofleft
\global\dimen1=\shortenproofright
\global\dimen2=\proofrulebreadth
\global\dimen3=\proofbelowshift
\global\dimen4=\proofdotseparation
\global\count255=\proofdotnumber
$\egroup  
\shortenproofleft=\dimen0
\shortenproofright=\dimen1
\proofrulebreadth=\dimen2
\proofbelowshift=\dimen3
\proofdotseparation=\dimen4
\proofdotnumber=\count255
}

\def\proofover{
\eproofbit 
\setbox\proofbelow=\hbox\bgroup 
\let\wereinproofbit\proofover
$\displaystyle
}%
\def\proofoverdbl{
\eproofbit 
\doubleprooftrue
\setbox\proofbelow=\hbox\bgroup 
\let\wereinproofbit\proofoverdbl
$\displaystyle
}%
\def\proofoverdots{
\eproofbit 
\proofdotstrue
\setbox\proofbelow=\hbox\bgroup 
\let\wereinproofbit\proofoverdots
$\displaystyle
}%
\def\proofusing{
\eproofbit 
\setbox\proofrulename=\hbox\bgroup 
\let\wereinproofbit\proofusing
\kern0.3em$
}

\def\endprooftree{
\eproofbit 
  \dimen5 =0pt
\dimen0=\wd\proofabove \advance\dimen0-\shortenproofleft
\advance\dimen0-\shortenproofright
\dimen1=.5\dimen0 \advance\dimen1-.5\wd\proofbelow
\dimen4=\dimen1
\advance\dimen1\proofbelowshift \advance\dimen4-\proofbelowshift
\ifdim  \dimen1<0pt
\then   \advance\shortenproofleft\dimen1
        \advance\dimen0-\dimen1
        \dimen1=0pt
        \ifdim  \shortenproofleft<0pt
        \then   \setbox\proofabove=\hbox{%
                        \kern-\shortenproofleft\unhbox\proofabove}%
                \shortenproofleft=0pt
        \fi
\fi
\ifdim  \dimen4<0pt
\then   \advance\shortenproofright\dimen4
        \advance\dimen0-\dimen4
        \dimen4=0pt
\fi
\ifdim  \shortenproofright<\wd\proofrulename
\then   \shortenproofright=\wd\proofrulename
\fi
\dimen2=\shortenproofleft \advance\dimen2 by\dimen1
\dimen3=\shortenproofright\advance\dimen3 by\dimen4
\ifproofdots
\then
        \dimen6=\shortenproofleft \advance\dimen6 .5\dimen0
        \setbox1=\vbox to\proofdotseparation{\vss\hbox{$\cdot$}\vss}%
        \setbox0=\hbox{%
                \advance\dimen6-.5\wd1
                \kern\dimen6
                $\vcenter to\proofdotnumber\proofdotseparation
                        {\leaders\box1\vfill}$%
                \unhbox\proofrulename}%
\else   \dimen6=\fontdimen22\the\textfont2 
        \dimen7=\dimen6
        \advance\dimen6by.5\proofrulebreadth
        \advance\dimen7by-.5\proofrulebreadth
        \setbox0=\hbox{%
                \kern\shortenproofleft
                \ifdoubleproof
                \then   \hbox to\dimen0{%
                        $\mathsurround0pt\mathord=\mkern-6mu%
                        \cleaders\hbox{$\mkern-2mu=\mkern-2mu$}\hfill
                        \mkern-6mu\mathord=$}%
                \else   \vrule height\dimen6 depth-\dimen7 width\dimen0
                \fi
                \unhbox\proofrulename}%
        \ht0=\dimen6 \dp0=-\dimen7
\fi
\let\doll\relax
\ifwasinsideprooftree
\then   \let\VBOX\vbox
\else   \ifmmode\else$\let\doll=$\fi
        \let\VBOX\vcenter
\fi
\VBOX   {\baselineskip\proofrulebaseline \lineskip.2ex
        \expandafter\lineskiplimit\ifproofdots0ex\else-0.6ex\fi
        \hbox   spread\dimen5   {\hfi\unhbox\proofabove\hfi}%
        \hbox{\box0}%
        \hbox   {\kern\dimen2 \box\proofbelow}}\doll%
\global\dimen2=\dimen2
\global\dimen3=\dimen3
\egroup 
\ifonleftofproofrule
\then   \shortenproofleft=\dimen2
\fi
\shortenproofright=\dimen3
\onleftofproofrulefalse
\ifinsideprooftree
\then   \hskip.5em plus 1fil \penalty2
\fi
}

\makeatother
 \newdir{ >}{{}*!/-7.5pt/@{>}}
 \newdir{|>}{!/4.5pt/@{|}*:(1,-.2)@^{>}*:(1,+.2)@_{>}}
 \newdir{ |>}{{}*!/-3pt/@{|}*!/-7.5pt/:(1,-.2)@^{>}*!/-7.5pt/:(1,+.2)@_{>}}
\newcommand{\xyline}[2][]{\ensuremath{\smash{\xymatrix@1#1{#2}}}}
\newcommand{\xyinline}[2][]{\ensuremath{\smash{\xymatrix@1#1{#2}}}^{\rule[8.5pt]{0pt}{0pt}}}
\makeatletter

\newif\ifignore 
\ignorefalse
\newcommand{\auxproof}[1]{
\ifignore\mbox{}\newline
\textbf{PROOF:} \dotfill\newline
{\it #1}\mbox{}\newline
\textbf{ENDPROOF}\dotfill
\fi}

\newtheorem{theorem}{Theorem}
\newtheorem{lemma}[theorem]{Lemma}
\newtheorem{proposition}[theorem]{Proposition}
\newtheorem{corollary}[theorem]{Corollary}
\newtheorem{definition}[theorem]{Definition}
\newtheorem{example}[theorem]{Example}

\newenvironment{proof}[1][Proof]%
   { \begin{trivlist}%
     \item[\hskip \labelsep {\bfseries #1}]%
   }%
   { \end{trivlist}%
   }
\newcommand{\QEDbox}{\square}
\newcommand{\QED}{\hspace*{\fill}$\QEDbox$}

\newcommand{\after}{\mathrel{\circ}}
\newcommand{\cat}[1]{\ensuremath{\mathbf{#1}}}
\newcommand{\Cat}[1]{\ensuremath{\mathbf{#1}}}

\newcommand{\op}{\ensuremath{^{\mathrm{op}}}}
\newcommand{\idmap}[1][]{\ensuremath{\mathrm{id}_{#1}}}
\newcommand{\id}[1][]{\idmap[#1]}
\renewcommand{\Im}{\ensuremath{\mathrm{Im}}}

\newcommand{\coker}{\ensuremath{\mathrm{coker}}}
\newcommand{\charac}{\ensuremath{\mathrm{char}}}

\newcommand{\KSub}{\ensuremath{\mathrm{KSub}}}

\newcommand{\sasaki}{\mathrel{\supset}}
\newcommand{\andthen}{\mathrel{\&}}
\newcommand{\powerset}{\mathcal{P}}
\newcommand{\nul}{\ensuremath{\underline{0}}}

\newcommand{\Rel}{\Cat{Rel}\xspace}
\newcommand{\PInj}{\Cat{PInj}\xspace}
\newcommand{\Hilb}{\Cat{Hilb}\xspace}
\newcommand{\PHilb}{\Cat{PHilb}\xspace}
\newcommand{\Sets}{\Cat{Sets}\xspace}

\newcommand{\cotuple}[2]{\ensuremath{[ #1,\,#2 ]}}
\newcommand{\Karoubi}[1]{{\cal K}(#1)}
\newcommand{\dagKaroubi}[1]{{\cal K}^{\dag}(#1)}

\newcommand{\set}[2]{\{#1\;|\;#2\}}
\newcommand{\setin}[3]{\{#1\in#2\;|\;#3\}}
\newcommand{\conjun}{\mathrel{\wedge}}
\newcommand{\disjun}{\mathrel{\vee}}

\newcommand{\allin}[3]{\forall_{#1\in#2}.\,#3}
\newcommand{\ex}[2]{\exists_{#1}.\,#2}
\newcommand{\exin}[3]{\exists_{#1\in#2}.\,#3}

\newcommand{\downset}{\mathop{\downarrow}\!}
\newcommand{\Perp}{\mathop{\perp}}
\newcommand{\effect}[1]{\mathfrak{E}(#1)}
\newcommand{\sai}[1]{[\,#1\,]}
\newcommand{\EndoHom}[1]{{\cal E}{\kern-.5ex}\textit{n}{\kern-.2ex}\textit{d}{\kern-.2ex}\textit{o}(#1)}

\title{Orthomodular lattices, Foulis Semigroups \\
       and Dagger Kernel Categories}
\author{Bart Jacobs, \\
{\small Institute for Computing and Information Sciences (iCIS),} \\[-.5em]
{\small Radboud University Nijmegen, The Netherlands.} \\[-.5em]
{\small Webaddress: \url{www.cs.ru.nl/B.Jacobs}}}

\date{\small \today}

\newenvironment{bijectivecorrespondence}
  {\newif\ifbijnotfirst
   \global\bijnotfirstfalse
   \global\def\bijprev{}
   \normalsize
   \begin{tabular}{cl}}
  {\end{tabular}
   }
\newcommand{\correspondence}[2][]{%
  \ifbijnotfirst%
    \rule{0pt}{5.8pt}%
    \smash{\ensuremath{\infer={\hphantom{#2}}{\hphantom{\bijprev}}}} \\%
  \fi%
  \global\bijnotfirsttrue%
  \global\def\bijprev{#2}%
  \ensuremath{#2} & #1 \\%
  }

\begin{document}
\maketitle

\begin{abstract}
This paper is a sequel to~\cite{HeunenJ09a} and continues the study of
quantum logic via dagger kernel categories. It develops the relation
between these categories and both orthomodular lattices and Foulis
semigroups. The relation between the latter two notions has been
uncovered in the 1960s. The current categorical perspective gives a
broader context and reconstructs this relationship between
orthomodular lattices and Foulis semigroups as special instance.
\end{abstract}

\section{Introduction}\label{IntroSec}

Dagger kernel categories have been introduced in~\cite{HeunenJ09a} as
a relatively simple setting in which to study categorical quantum
logic. These categories turn out to have orthomodular logic built in,
via their posets $\KSub(X)$ of kernel subobjects that can be used to
interprete predicates on $X$.  The present paper continues the study
of dagger kernel categories, especially in relation to orthomodular
lattices and Foulis semigroups. The latter two notions have been
studied extensively in the context of quantum logic. The main results
of this paper are as follows.
\begin{enumerate}
\item[(1)] A special category \Cat{OMLatGal} is defined with
  orthomodular lattices as objects and Galois connections between them
  as morphisms; it is shown that:
\begin{itemize}
\item \Cat{OMLatGal} is itself a dagger kernel category---with some
additional structure such a dagger biproducts, and an opclassifier;

\item for each dagger kernel category \Cat{D} there is a functor
  $\Cat{D} \rightarrow \Cat{OMLatGal}$ preserving the dagger kernel
  structure; hence \Cat{OMLatGal} contains in a sense all dagger
  kernel categories.
\end{itemize}

\item[(2)] For each object $X$ in a dagger kernel category, the homset
  $\EndoHom{X} = \Cat{D}(X,X)$ of endo-maps is a Foulis semigroup.

\item[(3)] Every Foulis semigroup $S$ yields a dagger kernel category
  $\dagKaroubi{S}$ via the ``dagger Karoubi'' construction
  $\dagKaroubi{-}$.
\end{enumerate}

Translations between orthomodular lattices and Foulis semigroups have
been described in the 1960s, see
\textit{e.g.}~\cite{Foulis60,Foulis62,Foulis63,BlythJ72,Kalmbach83}. These
translations appear as special instances of the above results:
\begin{itemize}
\item given a Foulis semigroup $S$, all the kernel posets $\KSub(s)$
  are orthomodular lattices, for each object $s\in\dagKaroubi{S}$ of
  the associated dagger kernel category (using point~(3) mentioned
  above). For the unit element $s=1$ this yields the ``old''
  translation from Foulis semigroups to orthomodular lattices;

\item given an orthomodular lattice $X$, the set of (Galois) endomaps
  $\EndoHom{X} = \Cat{OMLatGal}(X, X)$ on $X$ in the dagger kernel
  category \Cat{OMLatGal} forms a Foulis semigroup---using points~(1)
  and~(2). Again this is the ``old'' translation, from orthomodular
  lattices to Foulis semigroups.
\end{itemize}

\noindent Since dagger kernel categories are essential in these
constructions we see (further) evidence of the relevance of categories
in general, and of dagger kernel categories in particular, in the
setting of quantum (logical) structures.

The paper is organised as follows. Section~\ref{DagKerSec} first
recalls the essentials about dagger kernel categories
from~\cite{HeunenJ09a} and also about the (dagger) Karoubi envelope.
It shows that dagger kernel categories are closed under this
construction. Section~\ref{OMLatDagKerSec} introduces the fundamental
category \Cat{OMLatGal} of orthomodular lattices with Galois
connections between them, investigates some of its properties, and
introduces the functor $\KSub\colon \cat{D}\rightarrow \Cat{OMLatGal}$
for any dagger kernel category \Cat{D}. Subsequently,
Section~\ref{FoulisDagKerSec} recalls the definition of Foulis semigroups,
shows how they arise as endo-homsets in dagger kernel categories,
and proves that their dagger Karoubi envelope yields a dagger
kernel category. The paper ends with some final remarks and further
questions in Section~\ref{ConclusionSec}.

\section{Dagger kernel categories}\label{DagKerSec}

Since the notion of dagger kernel category is fundamental in this
paper we recall the essentials from~\cite{HeunenJ09a}, where this type
of category is introduced. Further details can be found there.

A dagger kernel category consists of a category \Cat{D} with a dagger
functor $\dag \colon \Cat{D}\op \rightarrow \Cat{D}$, a zero object
$0\in\Cat{D}$, and dagger kernels. The functor $\dag$ is the identity
on objects $X\in\Cat{D}$ and satisfies $f^{\dag\dag} = f$ on morphisms
$f$. The zero object $0$ yields a zero map, also written as $0$,
namely $X\rightarrow 0 \rightarrow Y$ between any two objects
$X,Y\in\Cat{D}$. A dagger kernel of a map $f\colon X\rightarrow Y$ is
a kernel map, written as $\smash{\xymatrix@C-1pc{k\colon K\ar@{
      |>->}[r] & X}}\!$, which is---or can be chosen as---a dagger mono,
meaning that $k^{\dag} \after k = \idmap[K]$. Often we write $k =
\ker(f)$, and $\coker(f) = \ker(f^{\dag})^{\dag}$ for the cokernel of
$f$. The definition $k^{\perp} = \ker(k^{\dag})$ for a kernel $k$
yields an orthocomplement.

We write \Cat{DKC} for the category with dagger kernel categories
as objects and functors preserving $\dag, 0, \ker$.

The main examples of dagger kernel categories are: \Rel, the category
of sets and relations, its subcategory \Cat{pInj} of sets and partial
injections, \Hilb, the category of Hilbert spaces and
bounded/continuous linear maps between them, and \PHilb, the category
of projective Hilbert spaces. This paper adds another example, namely
\Cat{OMLatGal}.

The main results from~\cite{HeunenJ09a} about dagger kernel categories
are as follows.
\begin{enumerate}
\item Each poset $\KSub(X)$ of kernel subobjects of an object $X$ is
  an orthomodular lattice; this is the basis of the relevance of
  dagger kernel categories to quantum logic.

\item Pullbacks of kernels exist along arbitrary maps $f\colon
  X\rightarrow Y$, yielding a pullback (or substitution) functor
  $f^{-1} \colon \KSub(Y) \rightarrow \KSub(X)$. Explicitly, like
  in~\cite{Freyd64}, $f^{-1}(n) = \ker(\coker(n) \after f)$.

\item This pullback functor $f^{-1}$ has a left adjoint $\exists_{f}
  \colon \KSub(X)\rightarrow \KSub(Y)$, corresponding to image
  factorisation. These $f^{-1}$ and $\exists_f$ only preserve part of
  the logical structure---meets are preserved by $f^{-1}$ and joins by
  $\exists_f$, via the adjointness---but for instance negations and
  joins are not preserved by substitution $f^{-1}$, unlike what is
  standard in categorical logic, see \textit{e.g.}~\cite{Jacobs99a}.

  Substitution $f^{-1}$ and existential quantification $\exists_{f}$
  are inter-expressible, via $f^{-1}(m)^{\perp} =
  \exists_{f^\dag}(m^{\perp})$.

\item The logical ``Sasaki'' hook $\sasaki$ and ``and-then''
  $\andthen$ connectives---together with the standard adjunction
  between them~\cite{Finch70,CoeckeS04}---arise via this adjunction
  $\exists_{f} \dashv f^{-1}$, namely for $m,n,k\in \KSub(X)$ as:
$$\begin{array}{rclcrcl}
m \sasaki n 
& = &
\effect{m}^{-1}(n)
& \quad &
k \andthen m
& = &
\exists_{\effect{m}}(k) \\
& = &
m^{\perp} \disjun (m\conjun n)
& &
& = &
m \conjun (m^{\perp} \disjun k),
\end{array}$$

\noindent where $\effect{m} = m \after m^{\dag}\colon X\rightarrow X$
is the effect (see~\cite{DvurecenskijP00}) associated with the kernel
$m$.
\end{enumerate}

\subsection{Karoubi envelope}\label{KaroubiSubsec}

Next we recall the essentials of the so-called Karoubi envelope
(see~\cite{Karoubi78} or~\cite[Chapter~2, Exercise~B]{Freyd64})
construction---and its ``dagger'' version---involving the free
addition of splittings of idempotents to a category. The construction
will be used in Section~\ref{FoulisDagKerSec} to construct a dagger
kernel category out of a Foulis semigroup. It is thus instrumental,
and not studied in its own right.

An idempotent in a category is an endomap $s\colon X\rightarrow X$
satisfying $s\after s = s$. A splitting of such an idempotent $s$ is a
pair of maps $e\colon X\rightarrow Y$ and $m\colon Y\rightarrow X$
with $m\after e = s$ and $e\after m = \idmap[Y]$. Clearly, $m$ is then
a mono and $e$ is an epi. Such a splitting, if it exists, is unique
up-to-isomorphism.

\auxproof{
Assume also $s = m' \after e'\colon X\rightarrow Y' \rightarrow X$
with $e'\after m' = \idmap[Y']$. Then $m\after e = m' \after e'$,
so that:
$$m\after e \after m' = m'
\qquad\mbox{and}\qquad
m'\after e' \after m = m.$$

\noindent Hence the maps $\varphi = e'\after m\colon Y\rightarrow Y'$
and $\psi = e\after m' \colon Y'\rightarrow Y$ are each other's
inverses:
$$\begin{array}{rcl}
\varphi \after \psi
& = &
e'\after m \after e \after m' \\
& = &
e'\after m' \\
& = &
\idmap \\
\psi \after \varphi 
& = &
e \after m' \after e' \after m \\
& = &
e \after m \\
& = &
\idmap.
\end{array}$$
}

For an arbitrary category \Cat{C} the so-called Karoubi envelope
$\Karoubi{\Cat{C}}$ has idempotents $s\colon X\rightarrow X$ in
\Cat{C} as objects. A morphism $\smash{(X\stackrel{s}{\rightarrow}X)
  \,\stackrel{f}{\longrightarrow}\, (Y\stackrel{t}{\rightarrow}Y)}$ in
$\Karoubi{\Cat{C}}$ consists of a map $f\colon X\rightarrow Y$ in
\Cat{C} with $f\after s = f = t\after f$. The identity on an object
$(X,s)\in\Karoubi{\Cat{C}}$ is the map $s$ itself. Composition in
$\Karoubi{\Cat{C}}$ is as in \Cat{C}. The mapping $X\mapsto
(X,\idmap[X])$ thus yields a full and faithful functor ${\cal I}\colon
\Cat{C} \rightarrow \Karoubi{\Cat{C}}$.

The Karoubi envelope $\Karoubi{\Cat{C}}$ can be understood as the free
completion of \Cat{C} with splittings of idempotents. Indeed, an
idempotent $f\colon (X,s)\rightarrow (X,s)$ in $\Karoubi{\Cat{C}}$ can
be split as $f = \big((X,s) \,\smash{\stackrel{f}{\rightarrow}}\,
(X,f) \,\smash{\stackrel{f}{\rightarrow}}\, (X,s)\big)$. If $F\colon
\Cat{C}\rightarrow \Cat{D}$ is a functor to a category \Cat{D} in
which endomorphisms split, then there is an up-to-isomorphism unique
functor $\overline{F}\colon \Karoubi{\Cat{C}}\rightarrow \Cat{D}$ with
$\overline{F}\after{\cal I} \cong F$.

\auxproof{
For an object $(X,s)$, define $\overline{F}(X)$ via the splitting
in \Cat{D}:
$$\xymatrix@R-1pc{
FX\ar[rr]^-{Fs}\ar@{->>}[dr]_{e_s} & & FX \\
& \overline{F}X\ar@{ >->}[ur]_{m_s}
}$$

\noindent For $f\colon (X,s)\rightarrow (Y,t)$ we take
$\overline{F}(f) = e_{t} \after F(f) \after m_{s}$.  Then:
$$\begin{array}{rcl}
\overline{F}(\idmap[(X,s)])
& = &
e_{s} \after F(s) \after m_{s} \\
& = &
e_{s} \after m_{s} \after e_{s} \after m_{s} \\
& = &
\idmap \after \idmap \\
& = &
\idmap \\
\overline{F}(g) \after \overline{F}(f) 
& = &
e_{r} \after F(g) \after m_{t} \after e_{t} \after F(f) \after m_{s} \\
& = &
e_{r} \after F(g) \after F(t) \after F(f) \after m_{s} \\
& = &
e_{r} \after F(g \after t\after f) \after m_{s} \\
& = &
e_{r} \after F(g \after f) \after m_{s} \\
& = &
\overline{F}(g\after f).
\end{array}$$

\noindent This $\overline{F}$ is unique up-to-isomorphism: each
idempotent $s\colon X\rightarrow X$ can be understood as a 
splitting in $\Karoubi{\Cat{C}}$, namely of:
$$\xymatrix@R-1pc{
{\cal{I}}(X) = (X,\idmap)\ar[rr]^-{{\cal{I}}(s) = s}\ar@{->>}[dr]_{s} & & 
   (X,\idmap{}) = {\cal{I}}(X) \\
& (X,s)\ar@{ >->}[ur]_{s}
}$$

\noindent Since $\overline{F}$ preserves splittings and satisfies
$\overline{F} \after {\cal I} \cong F$, we get the definition as
described above.
}

Hayashi~\cite{Hayashi85} (see also~\cite{HoofmanM95}) has developed a
theory of semi-functors and semi-adjunctions that can be used to
capture non-extensional features, without uniques of mediating maps,
like for of exponents $\Rightarrow$~\cite{Scott80a,LambekS86},
products $\prod$~\cite{Jacobs91b}, or exponentials
$!$~\cite{Hoofman92}. The Karoubi envelope can be used to turn such
``semi'' notions into proper (extensional) ones. This also happens in
Section~\ref{FoulisDagKerSec}.

\smallskip

Now assume \Cat{D} is a dagger category. An endomap $s\colon X
\rightarrow X$ in \Cat{D} is called a self-adjoint idempotent if
$s^{\dag} = s = s\after s$. A splitting of such an $s$ consists, as
before, of maps $m,e$ with $m\after e = s$ and $e\after m = \idmap$.
In that case $e^{\dag}, m^{\dag}$ is also a splitting of $s$, so that
we get an isomorphism $\varphi = m^{\dag} \after m$ in a commuting
diagram:
$$\xymatrix@R-1pc@C+1pc{
X\ar[dd]_{s}\ar[dr]_{e}\ar@/^2ex/[drr]^-{m^{\dag}} \\
& Y\ar[dl]_{m}\ar[r]^(0.4){\varphi}_(0.4){\cong} 
   & Y\ar@/^2ex/[dll]^-{e^{\dag}} \\
X
}$$

\noindent Hence $e^{\dag} = m$, as subobjects, and $m^{\dag} = e$ as
quotients.

The dagger Karoubi envelope $\dagKaroubi{\Cat{D}}$ of $\Cat{D}$ is the
full subcategory of $\Karoubi{\Cat{D}}$ with self-adjoint idempotents
as objects, see also~\cite{Selinger08}. This is again a dagger
category, since if $f\colon (X,s) \rightarrow (Y,t)$ in
$\dagKaroubi{\Cat{D}}$, then $f^{\dag} \colon (Y,t) \rightarrow (X,s)$
because:
$$s \after f^{\dag}
=
s^{\dag} \after f^{\dag}
=
(f \after s)^{\dag}
=
f^{\dag}.$$

\noindent Similarly $f^{\dag} \after t = f^{\dag}$.  The functor
${\cal I}\colon \Cat{D} \rightarrow \Karoubi{\Cat{D}}$ factors via
$\dagKaroubi{\Cat{D}} \rightarrow \Karoubi{\Cat{D}}$. One can
understand $\dagKaroubi{\Cat{D}}$ as the free completion of \Cat{D}
with splittings of self-adjoint idempotents.

Selinger~\cite{Selinger08} shows that the dagger Karoubi envelope
construction $\dagKaroubi{-}$ preserves dagger biproducts and dagger
compact closedness. Here we extend this with dagger kernels in the
next result. It will not be used in the sequel but is included to show
that the dagger Karoubi envelope is quite natural in the current
setting.

\begin{proposition}
\label{KaroubiKernelProp}
If \Cat{D} is a dagger kernel category, then so is
$\dagKaroubi{\Cat{D}}$. Moreover, the embedding ${\cal I}\colon
\Cat{D} \rightarrow \dagKaroubi{\Cat{D}}$ is a map of dagger kernel
categories.
\end{proposition}

\begin{proof}
For each object $X\in\Cat{D}$, the zero map $0\colon X\rightarrow X$
is a zero object in \Cat{D}, since there is precisely one map
$(X,0)\rightarrow (Y,t)$ in $\dagKaroubi{\Cat{D}}$, namely the zero
map $0\colon X\rightarrow Y$. As canonical choice we take the zero
object $0\in\Cat{D}$ with zero map $0 = \idmap[0]\colon 0\rightarrow
0$, which is in the range of ${\cal I}\colon \Cat{D} \rightarrow
\dagKaroubi{\Cat{D}}$.

For an arbitrary map $f\colon (X,s)\rightarrow (Y,t)$ in
$\dagKaroubi{\Cat{D}}$, let $k\colon K\rightarrowtail X$ be the kernel
of $f\colon X\rightarrow Y$ in \Cat{D}. We obtain a map $s'\colon
K\rightarrow K$, as in:
$$\xymatrix@C+1pc@R-1pc{
K\ar@{ |>->}[r]^-{k} & X\ar[r]^-{f} & Y \\
K\ar@{..>}[u]^{s'}\ar[r]_-{k} & X\ar[u]_{s}
}$$

\noindent since $f\after s \after k = f\after k = 0$. We obtain that
$s'$ is a self-adjoint idempotent, using that $k$ is a dagger mono
(\textit{i.e.}~satisfies $k^{\dag} \after k = \idmap{}$).
$$\begin{array}{ccc}
\begin{array}{rcl}
s' \after s
& = &
k^{\dag} \after k \after s' \after s' \\
& = &
k^{\dag} \after s \after k \after s' \\
& = &
k^{\dag} \after s\after s \after k \\
& = &
k^{\dag} \after s\after k \\
& = &
k^{\dag} \after k \after s' \\
& = &
s'.
\end{array}
& \qquad &
\begin{array}{rcl}
s'^{\dag}
& = &
s'^{\dag} \after k^{\dag} \after k \\
& = &
(k\after s')^{\dag} \after k \\
& = &
(s \after k)^{\dag} \after k \\
& = &
k^{\dag} \after s^{\dag} \after k \\
& = &
k^{\dag} \after s \after k \\
& = &
k^{\dag} \after k \after s' \\
& = &
s'.
\end{array}
\end{array}$$

\noindent This yields a dagger kernel in $\dagKaroubi{\Cat{D}}$,
$$\xymatrix@C+1pc{
(K,s')\ar[r]^-{s\after k} & (X,s)\ar[r]^-{f} & (Y,t)
}$$

\noindent since:
\begin{itemize}
\item $s\after k$ is a morphism in $\dagKaroubi{\Cat{D}}$: $s\after (s
  \after k) = s\after k$ and $(s\after k) \after s' = s\after s\after
  k = s\after k$;

\item $s\after k$ is a dagger mono:
$$\begin{array}{rcl}
(s\after k)^{\dag} \after (s\after k)
& = &
(k \after s')^{\dag} \after (k\after s') \\
& = &
s'^{\dag} \after k^{\dag} \after k \after s' \\
& = &
s' \after s' \\
& = &
s' \\
& = &
\idmap[(K,s')];
\end{array}$$

\item $f\after (s \after k) = f\after k = 0$;

\item if $g\colon (Z,r)\rightarrow (X,s)$ satisfies $f\after g = 0$,
  then there is a map $h\colon Z\rightarrow K$ in \Cat{D} with
  $k\after h = g$.  Then $s' \after h = h$, since $k\after s' \after h
  = s \after k \after h = s \after g = g = k \after h$. Similarly, $h
  \after r = h$, since $k\after h \after r = g \after r = g = k\after
  h$. Hence $h$ is a morphism $(Z,r)\rightarrow (K,s')$ in
  $\dagKaroubi{\Cat{D}}$ with $(s\after k) \after h = s\after g =
  g$. It is the unique one with this property since $s\after k$ is a
  (dagger) mono. \QED
\end{itemize}
\end{proof}

\begin{example}
\label{KaroubiKernelEx}
In the category \Hilb self-adjoint idempotents $s\colon H\rightarrow
H$ are also called projections. They can be written as $s = m \after
m^{\dag} = \effect{m}$ for a closed subspace $m\colon M\rightarrowtail
H$, see any textbook on Hilbert spaces
(\textit{e.g.}~\cite{Dvurecenskij92}). Hence they split already in
\Hilb, and so the dagger Karoubi envelope $\dagKaroubi{\Hilb}$ is
isomorphic to \Hilb: it does not add anything.

For the category \Rel of sets and relations the sitation is different.
A self-adjoint idempotent $S\colon X\rightarrow X$ is a relation
$S\subseteq X\times X$ that is both symmetric (since $S^{\dag} = S$)
and transitive (since $S\after S = S$), and thus a ``partial
equivalence relation'', commonly abbreviated as PER. The dagger
Karoubi envelope $\dagKaroubi{\Rel}$ has such PERs as objects. A
morphism $R \colon (S\subseteq X\times X) \rightarrow (T\subseteq
Y\times Y)$ is a relation $R\colon X\rightarrow Y$ with $R\after S = R
= T\after R$.  


\end{example}

Finally we note that the ``effect'' operation $\effect{m} = m \after
m^{\dag}$ can be described as a functor from the (total) category
$\KSub(\Cat{D})$ of kernels of a dagger kernel category
(see~\cite{HeunenJ09a}) to the dagger Karoubi envelope
$\dagKaroubi{\Cat{D}}$ as in the diagram:
$$\xymatrix{
\KSub(\Cat{D})\ar[rr]^-{\effect{-}} & & \dagKaroubi{\Cat{D}}
}$$

\noindent via:
$$\xymatrix@R-2.5pc{
M\ar@{ |>->}[dd]_{m}\ar@{-->}[r]^-{\varphi} & N\ar@{ |>->}[dd]^{n}
\\
& & \longmapsto &
\Big(X\ar[r]^-{\effect{m}} & X\Big)\ar[r]^-{f\after\effect{m}} &
   \Big(Y\ar[r]^-{\effect{n}} & Y\Big) \\
X\ar[r]^-{f} & Y 
}$$

\noindent We use that the necessarily unique map $\varphi\colon
M\rightarrow N$ with $n\after \varphi = f \after m$ satisfies
$\varphi = n^{\dag} \after n \after \varphi = n^{\dag} \after f \after m$.
Hence:
$$\effect{n} \after f \after \effect{m} 
= 
n \after n^{\dag} \after f \after m \after m^{\dag} 
= 
n \after \varphi \after m^{\dag} 
= 
f \after m \after m^{\dag} 
= 
f \after \effect{m},$$

\noindent so that $f \after \effect{m}$ is a morphism $\effect{m}\rightarrow
\effect{n}$ in the dagger Karoubi envelope $\dagKaroubi{\cat{D}}$. It
is not hard to see that this functor is full.

\auxproof{
We check functoriality:
$$\begin{array}{rcl}
\effect{\idmap[m]}
& = &
\idmap[X] \after \effect{m} \\
& = &
\effect{m} \\
& = &
\idmap[\effect{m}] \\
\effect(g\after f)
& = &
g \after f \after \effect{m} \\
& = &
g \after \effect{n} \after f \after \effect{m} \\
& = &
\effect{g} \after \effect{f}.
\end{array}$$

\noindent As to fulness, assume $g\colon \effect{m} \rightarrow
\effect{n}$ in $\dagKaroubi{\Cat{D}}$. Taking $\varphi = n^{\dag}
\after g \after m\colon M\rightarrow N$ we get:
$$n \after \varphi
=
n \after n^{\dag} \after g \after m
=
\effect{n} \after g \after m
=
g \after m.$$

\noindent Hence $g$ is a map $m\rightarrow n$ in $\KSub(\Cat{D})$.
It is mapped to itself: $\effect{g} = g \after \effect{m} = g$.
}

\section{Orthomodular lattices and Dagger kernel categories}\label{OMLatDagKerSec}

In~\cite{HeunenJ09a} it was shown how each dagger kernel category
gives rise to an indexed collection of orthomodular lattices, given by
the posets of the kernel subobjects $\KSub(X)$ of each object
$X$. Here we shall give a more systematic description of the situation
and see that a suitable category \Cat{OMLatGal} of orthomodular
lattices---with Galois connections between them---is itself a dagger
kernel category.  The mapping $\KSub(-)$ turns out to be functor to
this category \Cat{OMLatGal}, providing a form of representation of
dagger kernel categories.

We start by recalling the basic notion of orthomodular lattices.  They
may be understood as a non-distributive generalisation of Boolean
algebras. The orthomodularity formulation is due to~\cite{Husimi37},
following~\cite{BirkhoffN36}.

\begin{definition}
\label{OMLatDef}
A meet semi-lattice $(X,\conjun 1)$ is called an ortholattice if it
comes equipped with a function $(-)^{\perp}\colon X \to X$ satisfying:
\begin{itemize}
   \item $x^{\perp\perp} = x$;
   \item $x \leq y$ implies $y^\perp \leq x^\perp$;
   \item $x \conjun x^\perp = 0$.
\end{itemize}

\noindent One can then define a bottom element as $0 = 1 \conjun
1^{\perp} = 1^\perp$ and join by $x\disjun y = (x^{\perp}\conjun
y^{\perp})^{\perp}$, satisfying $x\disjun x^{\perp} = 1$.

Such an ortholattice is called orthomodular if it satisfies (one of)
the three equivalent conditions:
\begin{itemize}
\item $x \leq y$ implies $y = x \disjun (x^\perp \conjun y)$;

\item $x \leq y$ implies $x = y \conjun (y^\perp \disjun x)$;

\item $x \leq y$ and $x^{\perp} \conjun y = 0$ implies $x=y$.
\end{itemize}
\end{definition}

We shall consider two ways of organising orthomodular lattices
into a category.

\begin{definition}
The categories \Cat{OMLat} and \Cat{OMLatGal} both have orthomodular
lattices as objects.
\begin{enumerate}
\item A morphism $f\colon X\rightarrow Y$ in \Cat{OMLat} is a function
  $f\colon X\rightarrow Y$ between the underlying sets that preserves
  $\conjun, 1, (-)^{\perp}$---and thus also $\leq$, $\disjun$ and $0$;

\item A morphism $X\rightarrow Y$ in \Cat{OMLatGal} is a pair $f =
  (f_{*}, f^{*})$ of ``antitone'' functions $f_{*}\colon X\op
  \rightarrow Y$ and $f^{*}\colon Y \rightarrow X\op$ forming a Galois
  connection (or adjunction $f^{*}\dashv f_{*}$): $x\leq f^{*}(y)$ iff
  $y\leq f_{*}(x)$ for $x\in X$ and $y\in Y$.

The identity morphism on $X$ is the pair $(\bot,\bot)$ given by the
self-adjoint map $\idmap^{*} = \idmap[*] = (-)^{\perp} \colon X\op
\rightarrow X$. Composition of $\smash{X \stackrel{f}{\rightarrow} Y
  \stackrel{g}{\rightarrow} Z}$ is given by:
$$\begin{array}{rclcrcl}
(g\after f)_{*}
& = &
g_{*} \after \bot \after f_{*}
& \quad\mbox{and}\quad &
(g\after f)^{*}
& = &
f^{*} \after \bot \after g^{*}.
\end{array}$$
\end{enumerate}
\end{definition}



The category \Cat{OMLat} is the more obvious one, capturing the
(universal) algebraic notion of morphism as structure preserving
mapping. However, the category \Cat{OMLatGal} has more interesting
structure, as we shall see. It arises by restriction from a familiar
construction to obtain a (large) dagger category with involutive
categories as objects and adjunctions between them,
see~\cite{Heunen09}. The components $f_{*}\colon X\op \rightarrow Y$
and $f^{*}\colon Y \rightarrow X\op$ of a map $f\colon X\rightarrow Y$
in \Cat{OMLatGal} are not required to preserve any structure, but the
Galois connection yields that $f_{*}$ preserves meets, as right
adjoint, and thus sends joins in $X$ (meets in $X\op$) to meets in
$Y$, and dually, $f^{*}$ preserves joins and sends joins in $Y$ to
meets in $X$. 

The category $\Cat{OMLatGal}$ indeed has a dagger, namely by twisting:
$$\begin{array}{rcl}
(f_*,f^*)^{\dag}
& = &
(f^*,f_*).
\end{array}$$

\noindent A morphism $f\colon X \rightarrow Y$ in \Cat{OMLatGal} is a
dagger mono precisely when it safisfies $f^*(f_*(x)^\perp)=x^\perp$
for all $x \in X$, because $\id^*(x) = x^\perp = \id[*](x)$ and:
$$  (f^\dag \after f)^*(x) = f^*(f_*(x)^\perp) = (f^\dag \after f)_*(x).$$

In a Galois connection like $f^{*} \dashv f_{*}$ one map determines
the other. This standard result can be useful in proving equalities,
which, for convenience, we make explicit.

\begin{lemma}
\label{MapEqLem}
Suppose we have parallel maps $f,g\colon X\rightarrow Y$ in
\Cat{OMLatGal}. In order to prove $f=g$ it suffices to prove either
$f_{*} = g_{*}$ or $f^{*}=g^{*}$.
\end{lemma}

\begin{proof}
We shall prove that $f_{*} = g_{*}$ suffices to obtain also 
$f^{*}=g^{*}$. For all $x\in X$ and $y\in Y$,
$$x\leq f^{*}(y) \Longleftrightarrow
y \leq f_{*}(x) = g_{*}(x) \Longleftrightarrow
x \leq g^{*}(y).$$

\noindent Given $y$ this holds for all $x$, and so in particular for
$x = f^{*}(y)$ and $x=g^{*}(y)$, which yields $f^{*}(y) =
g^{*}(y)$. \QED
\end{proof}

Despite this result we sometimes explicitly write out both equations
$f_{*}=g_{*}$ and $f^{*}=g^{*}$, in particular when there is a special
argument involved.

The following elementary lemma is fundamental.

\begin{lemma}
\label{DownsetLem}
Let $X$ be an orthomodular lattice, with element $a\in X$. 
\begin{enumerate}
\item The (principal) downset $\downset a = \setin{u}{X}{u \leq a}$ is
  again an orthomodular lattice, with order, conjunctions and
  disjunctions as in $X$, but with its own orthocomplement $\perp_a$ given
  by $u^{\perp_a} = a \conjun u^{\perp}$, where $\perp$ is the
  orthocomplement from $X$.

\item There is a dagger mono $\downset a \rightarrowtail X$ in
  \Cat{OMLatGal}, for which we also write $a$, with
$$\begin{array}{rclcrcl}
a_{*}(u) & = & u^{\perp}
& \quad\mbox{and}\quad &
a^{*}(x) & = & a \conjun x^{\perp}.
\end{array}$$
\end{enumerate}
\end{lemma}

\begin{proof}
For the first point we check, for $u\in \downset a$,
$$u^{\perp_{a}\perp_{a}}
=
a \conjun (a \conjun u^{\perp})^{\perp} 
=
a \conjun (a^{\perp} \disjun u) 
=
u,$$

\noindent by orthomodularity, since $u\leq a$. We get a map in
\Cat{OMLatGal} because for arbitrary $u\in \downset a$ and $x\in X$,
$$x\leq a_{*}(u) = u^{\perp}
\Longleftrightarrow
u \leq x^{\perp}
\Longleftrightarrow
u \leq a \conjun x^{\perp} = a^{*}(x).$$

\noindent This map $a\colon \downset a\rightarrow X$ is a dagger mono
since:
$$a^{*}(a_{*}(u)^{\perp}) =
a^{*}(u^{\perp\perp}) =
a^{*}(u) =
a \conjun u^{\perp} =
u^{\perp_a}. \eqno{\QEDbox}$$
\end{proof}

We should emphasise that the equation $u^{\perp_{a}\perp_{a}} = u$ only
holds for $u\leq a$, and not for arbitrary elements $u$.

Later, in Proposition~\ref{DownsetIsKerProp}, we shall see that these
maps $\downset a \rightarrowtail X$ are precisely the kernels in the
category \Cat{OMLatGal}. But we first show that this category has
kernels in the first place.

To begin, $\Cat{OMLatGal}$ has a zero object $\nul$, namely the
one-element orthomodular lattice $\{*\}$. We can write its unique
element as $*=0=1$. Let us show that the lattice $\nul$ is indeed a
final object in $\Cat{OMLatGal}$. Let $X$ be an arbitrary orthomodular
lattice. The only function $f_* \colon X \to \nul$ is $f_*(x)=1$. It
has an obvious left adjoint $f^* \colon \nul \to L$ defined by $f^*(1)
= 1$:
$$\begin{prooftree}
x \;\leq\; f^{*}(1)\rlap{$\;=1$} 
\Justifies
1 \leq 1\rlap{$\;=f_{*}(x)$}
\end{prooftree}$$

\noindent Likewise, the unique morphism $g\colon \nul \to Y$ is given
by $g_*(1)=1$ and $g^*(y)=1$. Hence the zero morphism $z \colon X \to
Y$ is determined by $z_*(x)=1$ and $z^*(y)=1$.

\begin{theorem}
\label{OMLatGalDagKerCatThm}
  The category $\Cat{OMLatGal}$ is a dagger kernel category. The
  (dagger) kernel of a morphism $f \colon X \to Y$ is $k \colon\!
  \downset k \to X$, where $k = f^*(1) \in X$, like in
  Lemma~\ref{DownsetLem}.
\end{theorem}

\begin{proof}
The composition $f \after k$ is the zero map $\downset k \rightarrow Y$.
First, for $u\in \downset f^{*}(1)$,
$$(f \after k)_*(u) 
=
f_*(k_*(u)^\perp)
=
f_*(u)
= 
1,$$

\noindent because $u \leq f^*(1)$ in $X$ and so $1 \leq f_*(u)$ in
$Y$. And for $y\in Y$,
$$(f \after k)^*(y)
=
k^*(f^*(y)^\perp)
=
f^*(y) \conjun f^*(1)
=
f^*(y \disjun 1)
=
f^*(1)
=
k
=
1_{\downset k}.$$

\noindent because $f^{*}\colon Y\rightarrow X\op$ preserves joins as a
left adjoint.

Now suppose that $g \after k$ equals the zero morphism, for $g\colon Z
\to X$. Then $f_* \after \Perp \after g_*=1$ and $g^* \after \Perp
\after f^* = 1$. Hence for $z \in Z$ we have $1 \leq
f_*(g_*(z)^\perp)$, so $g_*(z)^\perp \leq f^*(1) = k$.  Define $h_*
\colon Z\op \to \downset k$ by $h_*(z) = g_*(z) \conjun k$, and define
$h^* \colon \downset k \to Z\op$ by $h^*(u) = g^*(u)$.  Then $h^*
\dashv h_*$ since for $u \leq k$ and $z \in Z$:
$$\begin{prooftree}
\begin{prooftree}
z \;\leq\; g^{*}(u) \rlap{$\;=h^{*}(u)$} \\
\Justifies
u \;\leq\; g_{*}(z)
\end{prooftree}
\Justifies
u \;\leq\; g_{*}(z) \conjun k \rlap{$\;=h_{*}(z)$}
\end{prooftree}$$

\noindent  whence $h$ is a well-defined morphism of $\Cat{OMLatGal}$.
It satisfies:
$$\begin{array}{rcl}
(k \after h)_*(z)
& = &
k_*(h_*(z)^{\perp_{\downset k}}) \\
& = &
k_*( (g_*(z) \conjun k )^{\perp_{\downset k}} ) \\
& = &
( (g_*(z) \conjun k )^{\perp_{\downset k}}  \conjun k )^\perp \\
& = &
( (g_*(z) \conjun k)^\perp \conjun k \conjun k)^\perp \\
& = &
(g_*(z) \conjun k ) \disjun k^\perp \\
& = &
g_*(z),
\end{array}$$

\noindent by orthomodularity since $k^{\perp} = f^*(1)^{\perp} \leq
g_*(z)$ because $g_*(z)^\perp \leq f^*(1) = k$ which follows from $1
\leq f_*(g_*(z)^\perp)$. Hence $h$ is a mediating morphism satisfying
$k \after h = g$.  It is the unique such morphism, since $k$ is a
(dagger) mono.  \QED

\auxproof{
In order to prove $(k\after h)^{*} = g^{*}$ we use
the following observation. Since $(f\after g)^{*} = 1$ we have
$g^{*}(f^{*}(y)^{\perp}) = 1$, for each $y\in Y$, and thus:
\[
\begin{array}{rcl}
g^{*}(x \disjun f^{*}(y)^{\perp})
& = &
g^{*}(x) \disjun_{Z\op} g^{*}(f^{*}(y)^{\perp}) \\
& &
   \qquad\mbox{since $g^{*}\colon X\rightarrow Z\op$ is a left adjoint} \\
& = &
g^{*}(x) \conjun_{Z} 1 \\
& = &
g^{*}(x).
\end{array}\eqno{(*)}
\]

\noindent We use this property $(*)$ twice below, as indicated:
$$\begin{array}{rcl}
(k \after h)^*(x) 
& = &
h^*(k^*(x)^{\perp_K}) \\
& = &
g^*( k^{*}(x)^{\perp} \conjun f^{*}(1) ) \\
& = &
g^*\big((x \disjun f^{*}(1)^{\perp}) \conjun f^{*}(1)\big) \\
& \smash{\stackrel{(*)}{=}} &
g^*\big(((x \disjun f^{*}(1)^{\perp}) \conjun f^{*}(1)) \disjun
             f^{*}(1)^{\perp}\big) \\
& = &
g^*(x \disjun f^{*}(1)^{\perp}) \\
& & \qquad
     \mbox{by orthomodularity, since } 
        x \disjun f^{*}(1)^{\perp} \geq f^{*}(1)^{\perp} \\
& \smash{\stackrel{(*)}{=}} &
g^{*}(x).
\end{array}$$
}
\end{proof}

For convenience we explicitly describe some of the basic structure
that results from dagger kernels, see~\cite{HeunenJ09a}, namely
cokernels and factorisations, given by dagger kernels and zero-epis.
We start with cokernels and zero-epis.

\begin{lemma}
\label{CokerLem}
The cokernel of a map $f\colon X\rightarrow Y$ in \Cat{OMLatGal} is:
$$\xymatrix@R-1pc{
\coker(f) = \Big(Y \ar@{-|>}[r]^-{c} & \;\downset f_{*}(1)\Big)
\qquad\mbox{with} &
{\left\{\begin{array}{rcl}
c_{*}(y) & = & y^{\perp} \conjun f_{*}(1) \\
c^{*}(v) & = & v^\perp \\
\end{array}\right.} 
}$$

\noindent Then:
$$\mbox{$f$ is zero-epi}
\;\;\smash{\stackrel{\textrm{def}}{\Longleftrightarrow}}\;\;
\coker(f) = 0
\;\Longleftrightarrow\;
f_{*}(1)=0.$$
\end{lemma}

\begin{proof}
Since:
$$\xymatrix{
\coker(f) =
\ker(f^{\dag})^{\dag} =
\big(\downset(f^{\dag})^{*}(1)\ar@{ |>->}[r] & Y\big)^{\dag} =
\big(Y \ar@{-|>}[r] & \downset f_{*}(1)\big).
}\eqno{\QEDbox}$$
\end{proof}

We recall from~\cite{HeunenJ09a} that each map $f$ in a dagger kernel
category has a zero-epi/kernel factorisation $f = i_{f} \after
e_{f}$. In combination with the factorisation of $f^{\dag}$ it yields
a factorisation $f = i_{f} \after m_{f} \after (i_{f^{\dag}})^{\dag}$ as in:
$$\xymatrix@R-1pc{
X\ar[rrr]^-{f}\ar@{-|>}[dr]_{(i_{f^\dag})^{\dag}} & & & Y \\
& \Im(f^{\dag})\ar@{ >->>}|{\circ}[r]_-{m_f} & \Im(f)\ar@{ |>->}[ur]_{i_f} &
}$$

\noindent where the map $m_f$ is both zero-epic and zero-monic, and
where $m_{f} \after (i_{f^{\dag}})^{\dag} = e_{f}$, the zero-epic part
of $f$.

\begin{lemma}
\label{ImFacLem}
For a map $f\colon X\rightarrow Y$ in \Cat{OMLatGal} one has:
$$\xymatrix@R-2pc@C-1pc{
\Big(\Im(f) = \downset (f_{*}(1)^{\perp}) \ar@{ |>->}[r]^-{i_f} & Y\Big)
\qquad\rlap{with} &
{\left\{\begin{array}{rcl}
(i_{f})_{*}(v) & = & v^{\perp} \\
(i_{f})^{*}(y) & = & y^{\perp} \conjun f_{*}(1)^{\perp} \\
\end{array}\right.} \\
\Big(X \ar@{->>}[r]^-{e_f}|-{\circ} & \downset f_{*}(1)^{\perp}\Big)
\quad\mbox{is} &
{\left\{\begin{array}{rcl}
(e_{f})_{*}(x) & = & f_{*}(x) \conjun f_{*}(1)^{\perp} \\
(e_{f})^{*}(v) & = & f^{*}(v)
\end{array}\right.} \\
\Big(\downset f^{*}(1)^{\perp} \ar@{ >->>}[r]^-{m_f}|-{\circ} & 
   \downset f_{*}(1)^{\perp}\Big)
\quad\mbox{is} &
{\left\{\begin{array}{rcl}
(m_{f})_{*}(x) & = & f_{*}(x) \conjun f_{*}(1)^{\perp} \\
(m_{f})^{*}(v) & = & f^{*}(v) \conjun f^{*}(1)^{\perp}
\end{array}\right.} \\
}$$
\end{lemma}

\begin{proof}
This is just a matter of unravelling definitions. For instance, 
$$\xymatrix{
\Im(f) = \ker(\coker(f)) = 
\ker\big(Y \ar@{-|>}[r]^-{c} & \downset f_{*}(1)\big) =
\big(\downset f_{*}(1)^{\perp} \;\ar@{ |>->}[r] & Y\big).
}$$

\noindent since $c^{*}(1_{\downset f_{*}(1)}) = c^{*}(f_{*}(1)) =
f_{*}(1)^{\perp}$. We check that $i_{f} \after e_{f} = f$, as
required.
$$\begin{array}{rcl}
(i_{f} \after e_{f})_{*}(x)
& = &
(i_{f})_{*}((e_{f})_{*}(x)^{\perp_{f_{*}(1)^{\perp}}}) \\
& = &
\big((e_{f})_{*}(x)^{\perp}\conjun f_{*}(1)^{\perp}\big)^{\perp} \\
& = &
(f_{*}(x) \conjun f_{*}(1)^{\perp}) \disjun f_{*}(1) \\
& = &
f_{*}(x),
\end{array}$$

\noindent by orthomodularity, using that $f_{*}(1) \leq f_{*}(x)$, 
since $x\leq 1$. This map $e_f$ is indeed zero-epic by the previous
lemma, since:
$$(e_{f})_{*}(1)
=
f_{*}(1) \conjun f_{*}(1)^{\perp}
=
0.$$

Next we first observe:
$$f_{*}(x \disjun f^{*}(1)) =
f_{*}(x) \conjun f_{*}(f^{*}(1)) =
f_{*}(x) \conjun 1 =
f_{*}(x),\eqno{(*)}$$

\noindent since there is a ``unit'' $1 \leq f_{*}(f^{*}(1))$. We use
this twice, in the marked equations, in:
$$\begin{array}[b]{rcl}
(m_{f} \after (i_{f^\dag})^{\dag})_{*}(x)
& = &
\big((m_{f})_{*} \after (-)^{\perp_{f^{*}(1)^{\perp}}} \after 
   (i_{f^{\dag}})^{*}\big)(x) \\
& = &
(m_{f})_{*}\big(f^{*}(1)^{\perp} \conjun 
   ((f^{\dag})_{*}(1)^{\perp} \conjun x^{\perp})^{\perp}\big) \\
& = &
f_{*}\big(f^{*}(1)^{\perp} \conjun (f^{*}(1) \disjun x)\big)
   \conjun f_{*}(1)^{\perp} \\
& \smash{\stackrel{(*)}{=}} &
f_{*}\big(f^{*}(1) \disjun (f^{*}(1)^{\perp} \conjun (f^{*}(1) \disjun x))\big)
   \conjun f_{*}(1)^{\perp} \\
& = &
f_{*}(f^{*}(1) \disjun x) \conjun f_{*}(1)^{\perp} \\
& \smash{\stackrel{(*)}{=}} &
f_{*}(x) \conjun f_{*}(1)^{\perp} \\
& = &
(e_{f})_{*}(x).
\end{array}$$

\noindent The map $m_f$ is zero-epic since:
$$\begin{array}{rcl}
(m_{f})_{*}(1_{\downset f^{*}(1)^{\perp}})
\hspace*{\arraycolsep} = \hspace*{\arraycolsep} 
(m_{f})_{*}(f^{*}(1)^{\perp})
& = &
f_{*}(f^{*}(1)^{\perp}) \conjun f_{*}(1)^{\perp} \\
& \smash{\stackrel{(*)}{=}} &
f_{*}(f^{*}(1) \disjun f^{*}(1)^{\perp}) \conjun f_{*}(1)^{\perp} \\
& = &
f_{*}(1) \conjun f_{*}(1)^{\perp} \\
& = &
0.
\end{array}$$

\noindent Similarly one shows that $m_f$ is zero-monic. \QED

\auxproof{
Similarly, one has
$$f^{*}(y \disjun f_{*}(1)) =
f^{*}(y) \conjun f^{*}(f_{*}(1)) =
f^{*}(y) \conjun 1 =
f^{*}(y),$$

\noindent since there is a ``counit'' $1 \leq f^{*}(f_{*}(1))$,
and thus:
$$\begin{array}[b]{rcl}
(i_{f} \after e_{f})^{*}(y)
& = &
f^{*}\big((i_{f})^{*}(y)^{\perp} \conjun f_{*}(1)^{\perp}\big) \\
& = &
f^{*}\big((y \disjun f_{*}(1)) \conjun f_{*}(1)^{\perp}\big) \\
& = &
f^{*}\big(((y \disjun f_{*}(1)) \conjun f_{*}(1)^{\perp}) \disjun f_{*}(1)\big) \\
& = &
f^{*}(y \disjun f_{*}(1)) \qquad \mbox{by orthomodularity} \\
& = &
f^{*}(y).
\end{array}$$

\noindent Via this auxiliary result one obtains that $m_f$ is zero-monic:
$$\begin{array}{rcl}
(m_{f})^{*}(1)
& = &
f^{*}(f_{*}(1)^{\perp}) \conjun f^{*}(1)^{\perp} \\
& = &
f^{*}(f_{*}(1) \disjun f_{*}(1)^{\perp}) \conjun f^{*}(1)^{\perp} \\
& = &
f^{*}(1) \conjun f^{*}(1)^{\perp} \\
& = &
0.
\end{array}$$
}
\end{proof}

For the record, inverse and direct images are described explicitly.

\begin{lemma}
\label{InvDirImLem}
For a map $f\colon X\rightarrow Y$ in \Cat{OMLatGal} the associated
inverse and direct images are:
$$\xymatrix@R-2pc{
\KSub(Y)\ar[r]^-{f^{-1}} & \KSub(X) 
&
  \KSub(X)\ar[r]^-{\exists_f} & \KSub(Y) \\
\big(\downset b\to Y)\ar@{|->}[r] &
   \big(\downset f^{*}(b^{\perp}) \to X\big)
&
\big(\downset a\to X)\ar@{|->}[r] &
   \big(\downset(f_{*}(a)^{\perp}) \to Y\big)
}$$
\end{lemma}

\begin{proof}
For $f\colon X\rightarrow Y$ and $b\in Y$, we have, using the
formulation for pullback of kernels from Section~\ref{DagKerSec}
(or~\cite[Lemma~2.4]{HeunenJ09a}) and Lemma~\ref{ImFacLem} above,
$$\begin{array}{rcll}
\lefteqn{f^{-1}(\downset b \rightarrow Y)} \\
& = &
\ker(\coker(\downset b\rightarrow Y) \after f) \\
& = &
\ker((Y\rightarrow \downset c) \after f), 
   & \mbox{for }
   c \begin{array}[t]{ll}
      = & b_{*}(1_{\downset b}) =  (1_{\downset b})^{\perp} = b^{\perp} 
     \end{array} \\
& = &
\downset a \rightarrow X, 
& \mbox{for } 
   a \begin{array}[t]{ll}
          = & (c^{\dag} \after f)^{*}(1_{\downset b})
            = f^{*}(c_{*}(1_{\downset c})^{\perp}) \\
          = & f^{*}(c^{\perp\perp}) = f^{*}(c) = f^{*}(b^{\perp})
        \end{array} \\
& = &
\downset f^{*}(b^{\perp}) \rightarrow X.
\end{array}$$

\noindent For $\exists_f$ we also use Lemma~\ref{ImFacLem}  in:
$$\begin{array}[b]{rcll}
\lefteqn{\exists_{f}(\downset a \rightarrow X)} \\
& = &
\Im(f \after (\downset a\rightarrow X)) \\
& = &
\downset b\rightarrow Y, & \mbox{where }
   b \begin{array}[t]{ll}
     = & (f\after a)_{*}(1_{\downset a})^{\perp} 
      = f_{*}(a_{*}(1_{\downset a})^{\perp})^{\perp} \\
      = & f_{*}(a^{\perp\perp})^{\perp}
      \end{array} \\
& = &
\downset(f_{*}(a)^{\perp}) \rightarrow Y.
\end{array}\eqno{\QEDbox}$$
\end{proof}

In the category \Cat{OMLatGal}, like in any dagger kernel category,
the kernel posets $\KSub(X)$ are orthomodular lattices. They turn out
to be isomorphic to the underlying object $X\in\Cat{OMLatGal}$.

\begin{proposition}
\label{DownsetIsKerProp}
Each dagger mono $a\colon\! \downset a\rightarrowtail X$ from
Lemma~\ref{DownsetLem}, for $a\in X$, is actually a dagger
kernel. This yields an isomorphism of orthomodular lattices
$$\xymatrix@C-.5pc{
X \ar[r]^-{\cong} & 
\KSub(X), \quad\mbox{namely}\quad a\ar@{|->}[r] & 
\Big(\downset a\ar[r]^-{a} & X\Big).
}$$

\noindent It is natural in the sense that for $f\colon X\rightarrow Y$
in \Cat{OMLatGal} the following squares commute by Lemma~\ref{InvDirImLem}.
$$\xymatrix@R-1pc{
X\ar[d]_{\cong}\ar[rr]^-{\bot \,\after\, f_{*}} & & 
   Y\ar[d]^{\cong}\ar[rr]^-{f^{*}\after \bot} & & X\ar[d]^{\cong} \\
\KSub(X)\ar[rr]_-{\exists_f} & & \KSub(Y)\ar[rr]_-{f^{-1}} & & \KSub(X)
}$$
\end{proposition}

\begin{proof}
  We first check that $a\colon\! \downset a \rightarrow X$ is indeed a
  kernel, namely of its cokernel $\coker(a)\colon X\rightarrow
  \downset a_{*}(1)$, see Lemma~\ref{CokerLem}, where $a_{*}(1) =
  a_{*}(1_{\downset a}) = a_{*}(a) = a^{\perp}$. Thus,
  $\ker(\coker(a)) = \coker(a)^{*}(1) = \coker(a)^{*}(1_{\downset
    a^{\perp}}) = \coker(a)^{*}(a^{\perp}) = a^{\perp\perp} = a$.

  Theorem~\ref{OMLatGalDagKerCatThm} says that the mapping $X
  \rightarrow \KSub(X)$ is surjective. Here we shall show that it is
  an injective homomorphism of orthomodular lattices reflecting the
  order, so that it is an isomorphism in the category \Cat{OMLat}.

Assume that $a\leq b$ in $X$. We can define $\varphi \colon \downset a
\rightarrow \downset b$ by $\varphi_{*}(x) = x^{\perp_{b}} = b \conjun
x^{\perp}$ and $\varphi^{*}(y) = a \conjun y^{\perp}$, for
$x\in\downset a$ and $y\in \downset b$. Then, clearly, $y\leq
\varphi_{*}(x)$ iff $x\leq \varphi^{*}(y)$, so that $\varphi$ is a
morphism in \Cat{OMLatGal}. In order to show $a \leq b$ in $\KSub(X)$
we prove $b\after \varphi = a$. First, for $x\in \downset a$,
$$\begin{array}{rcl}
(b \after \varphi)_{*}(x)
& = &
b_{*}\big(\varphi_{*}(x)^{\perp_{b}}\big) \\
& = &
b_{*}\big(x^{\perp_{b}\perp_{b}}\big) \\
& = &
b_{*}(x) \qquad\mbox{because $x\in\downset a\subseteq \downset b$} \\
& = &
x^{\perp} \\
& = &
a_{*}(x).
\end{array}$$

\auxproof{
Similarly, for an arbitrary $y\in X$,
$$\begin{array}{rcl}
(b \after \varphi)^{*}(y)
& = &
\varphi^{*}\big(b^{*}(y)^{\perp_{b}}\big) \\
& = &
\varphi^{*}\big(b \conjun (b \conjun y^{\perp})^{\perp}\big) \\
& = &
a \conjun \big(b \conjun (b \conjun y^{\perp})^{\perp}\big)^{\perp} \\
& = &
a \conjun (b^{\perp} \disjun (b \conjun y^{\perp})) \\
& \stackrel{*}{=} &
a \conjun y^{\perp} \\
& = &
a^{*}(y).
\end{array}$$

\noindent The marked equation holds because for arbitrary $z$,
$$\begin{array}{rcll}
a \conjun z
& = &
a \conjun (b\conjun z) & \mbox{because $a \leq b$} \\
& = &
a \conjun (b \conjun (b^{\perp} \disjun (b\conjun z)))
   \quad & \mbox{by orthomodularity, using $b\conjun z \leq b$} \\
& = &
a \conjun (b^{\perp} \disjun (b\conjun z)), 
   & \mbox{again because $a \leq b$.} \\
\end{array}$$
}

\noindent The map $X \rightarrow \KSub(X)$ does not only preserve the
order, but also reflects it: if we have an arbitrary map $\psi\colon
\downset a \rightarrow \downset b$ in \Cat{OMLatGal} with $b\after
\psi = a$, then:
$$\begin{array}{rcl}
a 
\hspace*{\arraycolsep} = \hspace*{\arraycolsep} 
a^{\perp\perp}
\hspace*{\arraycolsep} = \hspace*{\arraycolsep}
a_{*}(a)^{\perp}
& = &
(b \after \psi)_{*}(a)^{\perp} \\
& = &
b_{*}(\psi_{*}(a)^{\perp_{b}})^{\perp} \\
& = &
\psi_{*}(a)^{\perp_{b}\perp\perp} \\
& = &
\psi_{*}(a)^{\perp_{b}} \\
& = &
b \conjun \psi_{*}(a)^{\perp}
\hspace*{\arraycolsep} \leq \hspace*{\arraycolsep}
b.
\end{array}$$

\noindent This $X\rightarrow \KSub(X)$ map also preserves $\perp$,
since
$$\xymatrix{
\big(\downset a\ar@{ |>->}[r]^-{a} & X\big)^{\perp}
=
\ker(a^{\dag})
=
\big(\downset b\ar@{ |>->}[r]^-{b} & X\big)
}$$

\noindent where, according to Theorem~\ref{OMLatGalDagKerCatThm},
$b = (a^{\dag})^{*}(1_{\downset a}) = a_{*}(a) = a^{\perp}$.

It remains to show that the mapping $X\rightarrow \KSub(X)$ preserves
finite conjunctions. It is almost immediate that it sends the top
element $1\in X$ to the identity map (top) in $\KSub(X)$. It also
preserves finite conjunctions, since the intersection of the kernels
$\downset a\rightarrow X$ and $\downset b\rightarrow X$ is given by
$\downset(a\conjun b)\rightarrow X$. Since $a\conjun b \leq a, b$ there
are appropriate maps $\downset (a\conjun b) \rightarrow \downset a$
and $\downset (a\conjun b) \rightarrow \downset b$. Suppose that we
have maps $k \rightarrow \downset a$ and $k\rightarrow \downset b$,
where $k \colon \downset f^{*}(1)\rightarrow X$ is a kernel of
$f\colon X\rightarrow Y$. Since, as we have seen, the order is
reflected, we get $f^{*}(1) \leq a, b$, and thus $f^{*}(1) \leq
a\conjun b$, yielding the required map $\downset f^{*}(1) 
\rightarrow \downset (a\conjun b)$. \QED
\end{proof}

The adjunction $\exists_{f} \dashv f^{-1}$ that exists in arbitrary
dagger kernel categories (see Section~\ref{DagKerSec}
or~\cite[Proposition~4.3]{HeunenJ09a}) boils down in our example
\Cat{OMLatGal} to the adjunction between $f^{*} \dashv f_{*}$ in the
definition of morphisms in \Cat{OMLatGal}, since:
$$\begin{array}{rcl}
\exists_{f}(\downset a\rightarrow X) \leq (\downset b\rightarrow Y)
& \Longleftrightarrow &
(\downset f_{*}(a)^{\perp} \rightarrow Y) \leq (\downset b\rightarrow Y) \\
& \Longleftrightarrow &
f_{*}(a)^{\perp} \leq b \\
& \Longleftrightarrow &
b^{\perp} \leq f_{*}(a) \\
& \Longleftrightarrow &
a \leq f^{*}(b^{\perp}) \\
& \Longleftrightarrow &
(\downset a \rightarrow X) \leq (\downset f^{*}(b^{\perp}) \rightarrow X) \\
& \Longleftrightarrow &
(\downset a \rightarrow X) \leq f^{-1}(\downset b \rightarrow X).
\end{array}$$

\noindent Moreover, the Sasaki hook $\sasaki$ and and-then operators
$\andthen$ defined categorically
in~\cite[Proposition~6.1]{HeunenJ09a}, see Section~\ref{DagKerSec},
amount in \Cat{OMLatGal} to their usual definitions in the theory of
orthomodular lattices, see \textit{e.g.}~\cite{Finch70,Kalmbach83}. This will
be illustrated next.  We use the ``effect'' $\effect{\downset
a\rightarrow X} = a \after a^{\dag} \colon X\rightarrow X$ associated
with a kernel in:
$$(\downset a\rightarrow X) \sasaki (\downset b\rightarrow X)
\;\smash{\stackrel{\textrm{def}}{=}}\;
\effect{\downset a\rightarrow X}^{-1}(\downset b\rightarrow X)
=
(\downset c\rightarrow X),$$

\noindent where, according to the description of inverse image 
$(-)^{-1}$ in the previous lemma,
$$c 
=
(a \after a^{\dag})^{*}(b^{\perp})
=
a_{*}\big(a^{*}(b^{\perp})^{\perp_a}\big)
=
\big(a \conjun (a \conjun b)^{\perp}\big)^{\perp}
=
a^{\perp} \disjun (a\conjun b)
=
a \sasaki b.$$

\noindent Similarly for and-then $\andthen$:
$$(\downset a\rightarrow X) \andthen (\downset b\rightarrow X)
\;\smash{\stackrel{\textrm{def}}{=}}\;
\exists_{\effect{\downset b\rightarrow X}}(\downset a\rightarrow X) 
=
(\downset c\rightarrow X),$$

\noindent where the description of direct image from
Lemma~\ref{InvDirImLem} yields:
$$c
=
(b \after b^{\dag})_{*}(a)^{\perp}
=
b_{*}\big(b^{*}(a)^{\perp_b}\big)^{\perp}
=
\big(b \conjun (b \conjun a^{\perp})^{\perp}\big)^{\perp\perp}
=
b \conjun (b^{\perp} \disjun a)
=
a \andthen b.$$

\noindent These $\andthen$ and $\sasaki$ are, by construction, related
via an adjunction (see also~\cite{Finch70,CoeckeS04}).

Also one can define a weakest precondition modality $[f]$ from dynamic
logic in this setting: for $f\colon X\rightarrow Y$ and $y\in Y$, put:
$$\begin{array}{rcl}
[f](y) 
& \smash{\stackrel{\textrm{def}}{=}} &
f^{*}(y^{\perp}).
\end{array}$$

\noindent for ``$y$ holds after $f$''. This operation $[f](-)$
preserves conjunctions, as usual. An element $a\in X$ yields a test
operation $a? = \effect{a} = a \after a^{\dag}$. Then one can recover the
Sasaki hook $a \sasaki b$ via this modality as $[a?]b$, and hence
complement $a^\perp$ as $[a?]0$, see also
\textit{e.g.}~\cite{BaltagS06}.

There is another isomorphism of interest in this setting.

\begin{lemma}
\label{PointLem}
Let $2 = \{0,1\}$ be the 2-element Boolean algebra, considered as an
orthomodular lattice $2\in\Cat{OMLatGal}$. For each orthomodular lattice
$X$, there is an isomorphism (of sets): 
$$\xymatrix{
X \ar[r]^-{\cong} & \Cat{OMLatGal}(2, X)
}$$

\noindent which maps $a\in X$ to $\overline{a}\colon 2\rightarrow X$
given by:
$$\begin{array}{rclcrcl}
\overline{a}_{*}(w)
& = &
\left\{\begin{array}{ll}
1 & \mbox{if }w=0 \\
a^{\perp} & \mbox{if }w=1
\end{array}\right.
& \qquad &
\overline{a}^{*}(x)
& = &
\left\{\begin{array}{ll}
1 & \mbox{if }x\leq a^{\perp} \\
0 & \mbox{otherwise.}
\end{array}\right.
\end{array}$$

\noindent This isomorphism is natural: for $f\colon X\rightarrow Y$ one has:
$$\xymatrix@R-1pc{
X\ar[d]_{\cong}\ar[rr]^-{\bot\,\after\,f_{*}} & & Y\ar[d]^{\cong} \\
\Cat{OMLatGal}(2,X)\ar[rr]^-{f\after -} & & \Cat{OMLatGal}(2, Y)
}$$
\end{lemma}

\begin{proof}
The thing to note is that for a map $g\colon 2\rightarrow X$ in
\Cat{OMLatGal} we have $g_{*}(0) = 1$ because $g_{*}\colon 2\op
\rightarrow X$ is a right adjoint. Hence we can only choose $g_{*}(1)
\in X$.  Once this is chosen, the left adjoint $g^{*}\colon
X\rightarrow 2\op$ is completely determined, namely as $1 \leq
g^{*}(x)$ iff $x\leq g_{*}(1)$.

As to naturality, it suffices to show:
$$\begin{array}[b]{rcl}
(f \after \overline{a})_{*}(1)
& = &
f_{*}(\overline{a}_{*}(1)^{\perp}) \\
& = &
f_{*}(a^{\perp\perp}) \\
& = &
f_{*}(a)^{\perp\perp} \\
& = &
\overline{f_{*}(a)^{\perp}}_{*}(1) \\
& = &
\overline{(\bot\,\after\,f_{*})(a)}_{*}(1).
\end{array}\eqno{\QEDbox}$$
\end{proof}

By combining the previous two results we obtain a way to classify
(kernel) subobjects, like in a topos~\cite{MacLaneM92}, but with
naturality working in the opposite direction. In~\cite{HeunenJ09a} a
similar structure was found in the category \Rel of sets and
relations, and also in the dagger kernel category associated with a
Boolean algebra.

\begin{corollary}
\label{OpClassCor}
The 2-element lattice $2\in\Cat{OMLatGal}$ is an ``opclassifier'':
there is a ``characteristic'' isomorphism:
$$\xymatrix{
\KSub(X) \ar[r]^-{\charac}_-{\cong} & \Cat{OMLatGal}(2, X).
}$$

\noindent which is natural: $\charac \after \exists_{f} = f \after
\charac$. \QED
\end{corollary}

We conclude our investigation of the category \Cat{OMLatGal} with
the following observation.

\begin{proposition}
\label{BiprodProp}
The category $\Cat{OMLatGal}$ has (finite) dagger biproducts $\oplus$.
Explicitly, $X_1 \oplus X_2$ is the Cartesian product of (underlying
sets of) orthomodular lattices, with coprojection $\kappa_1\colon X_1
\to X_1 \oplus X_2$ defined by $(\kappa_1)_*(x) = (x^\perp,1)$ and
$(\kappa_1)^*(x,y) = x^\perp$. The dual product structure is given by
$\pi_i = (\kappa_i)^\dag$.
\end{proposition}

\begin{proof}
  Let us first verify that $\kappa_1$ is a well-defined morphism of
  $\Cat{OMLatGal}$, \textit{i.e.} that $(\kappa_1)^* \dashv (\kappa_1)_*$:
$$\begin{prooftree}
\begin{prooftree}
z \;\leq\; x^{\perp} \rlap{$\;=(\kappa_{1})^{*}(x,y)$}
\Justifies
x \;\leq\; z^{\perp} 
\end{prooftree}
\Justifies
(x,y) \;\leq\; (z^{\perp}, 1) \rlap{$\;=(\kappa_{1})_{*}(z)$}
\end{prooftree}$$

\noindent Also, $\kappa_1$ is a dagger mono since:
$$(\kappa_{1})^{*}\big((\kappa_{1})_{*}(x)^{\perp}\big) 
= 
(\kappa_{1})^{*}\big((x^{\perp}, 1)^{\perp}\big) 
= 
(\kappa_{1})^{*}\big(x,0\big) 
= 
x^{\perp}.$$ 

\noindent Likewise, there is a dagger mono $\kappa_2 \colon X_2 \to
X_1 \oplus X_2$. For $i \neq j$, one finds that $(\kappa_j)^\dag
\after \kappa_i$ is the zero morphism.

\auxproof{
$$\begin{array}{rcl}
\big((\kappa_{1})^{\dag} \after \kappa_{2}\big)_{*}(y)
& = &
(\kappa_{1})^{*}\big((\kappa_{2})_{*}(y)^{\perp}\big) \\
& = &
(\kappa_{1})^{*}\big((1, y^{\perp})^{\perp}\big) \\
& = &
(\kappa_{1})^{*}\big(0, y\big) \\
& = &
0^{\perp} \\
& = &
1.
\end{array}$$
}

In order to show that $X_1 \oplus X_2$ is indeed a coproduct, suppose
that morphisms $f_i \colon X_i \to Y$ are given.  We then define the
cotuple $\cotuple{f_1}{f_2} \colon X_1 \oplus X_2 \to Y$ by
$\cotuple{f_1}{f_2}_*(x_1,x_2) = (f_1)_*(x_1) \conjun (f_2)_*(x_2)$
and $\cotuple{f_1}{f_2}^*(y) = (f_1^*(y), f_2^*(y))$. Clearly,
$\cotuple{f_1}{f_2}^* \dashv \cotuple{f_1}{f_2}_*$, and:

\auxproof{
  \[
  \begin{bijectivecorrespondence}
    \correspondence[in $(X_1\oplus X_2)\op$]
       {\cotuple{f_1}{f_2}^*(y) = (f_1^*(y), f_2^*(y)) \leq (x_1,x_2)}
    \correspondence[in $X_i\op$]{f_i^*(y) \leq x_i}
    \correspondence[in $Y$]{y \leq (f_i)_*(x_i)}
    \correspondence[in $Y$.]{y \leq (f_1)_*(x_1) \conjun (f_2)_*(x_2) =
          \cotuple{f_1}{f_2}_*(x_1,x_2)}
  \end{bijectivecorrespondence}\]
}

$$\begin{array}{rcl}
(\cotuple{f_1}{f_2} \after \kappa_1)_*(x)
& = &
\cotuple{f_1}{f_2}_*\big((\kappa_1)_*(x)^\perp\big) \\
& = &
\cotuple{f_1}{f_2}_*\big((x^{\perp},1)^{\perp}\big) \\
& = &
(f_1)_*(x) \conjun (f_2)_*(0) \\
& = &
(f_1)_*(x) \conjun 1 \\
& = &
(f_1)_*(x).
\end{array}$$

\auxproof{
\begin{align*}
        (\cotuple{f_1}{f_2} \after \kappa_1)^*(y) 
    & = \kappa_1^*( \cotuple{f_1}{f_2}^*(y)^\perp ) \\
    & = \kappa_1^*( f_1^*(y)^\perp, f_2^*(y)^\perp ) \\
    & = f_1^*(y)^{\perp\perp} \\
    & = f_1^*(y),
  \end{align*}
}

\noindent so that $\cotuple{f_1}{f_2} \after \kappa_1 =
f_1$. Likewise, $\cotuple{f_1}{f_2} \after \kappa_2 = f_2$. Moreover,
if $g\colon X_{1}\oplus X_{2}\rightarrow Y$ also satisfies
$g\after\kappa_{i} = f_{i}$, then:
$$\begin{array}[b]{rcl}
\cotuple{f_1}{f_2}_{*}(x_{1},x_{2})
& = &
(f_1)_*(x_1) \conjun (f_2)_*(x_2) \\
& = &
g_{*}\big((\kappa_{1})_{*}(x_{1})^{\perp}\big) \conjun 
   g_{*}\big((\kappa_{2})_{*}(x_{2})^{\perp}\big) \\
& = &
g_{*}\big((x_{1}^{\perp},1)^{\perp}\big) \conjun 
   g_{*}\big((1,x_{2}^{\perp})^{\perp}\big) \\
& = &
g_{*}\big(x_{1},0\big) \conjun g_{*}\big(0,x_{2}\big) \\
& = &
g_{*}\big((x_{1},0) \disjun (0,x_{2})\big) \\
& = &
g_{*}\big(x_{1},x_{2}\big).
\end{array}\eqno{\QEDbox}$$
\end{proof}

\subsection{From dagger kernel categories to orthomodular lattices}

The aim in this subsection is to show that for a dagger kernel \Cat{D}
the kernel subobject functor $\KSub(-)$ is a functor $\cat{D}
\rightarrow \Cat{OMLatGal}$. On a morphism $f\colon X \to Y$ of
$\cat{D}$, define $\KSub(f) \colon \KSub(X) \to \KSub(Y)$ by:
$$\begin{array}{rclcrcl}
\KSub(f)_{*}(m) 
& = &
\big(\exists_{f}(m)\big)^{\perp}
& \qquad &
\KSub(f)^{*}(n)
& = &
f^{-1}\big(n^{\perp}\big).
\end{array}$$

\noindent Then indeed $\KSub(f)^* \dashv \KSub(f)_*$: 
$$\begin{prooftree}
n \;\leq\; (\exists_{f}(m))^{\perp} \rlap{$\;=\KSub(f)_{*}(m)$}
\Justifies
\begin{prooftree}
\exists_{f}(m) \;\leq\; n^{\perp}
\Justifies
m \;\leq\; f^{-1}(n^{\perp}) \rlap{$\;=\KSub(f)^{*}(n)$}
\end{prooftree}
\end{prooftree}$$

\auxproof{
As a result we get a functor $\KSub \colon \cat{D} \to
\Cat{OMLatGal}$, since:
$$\begin{array}{rcl}
\KSub(\idmap[X])_{*}(m)
& = &
\exists_{\idmap}(m)^{\perp} \\
& = &
m^{\perp} \\
& = &
(\idmap[\KSub(X)])_{*}(m) \\
\KSub(g\after f)_{*}(m) 
& = &
\exists_{g\after f}(m)^{\perp} \\
& = &
\exists_{g}(\exists_{f}(m))^{\perp} \\
& = &
\exists_{g}(\exists_{f}(m)^{\perp\perp})^{\perp} \\
& = &
\KSub(g)_{*}(\KSub(f)_{*}(m)^{\perp}) \\
& = &
(\KSub(g) \after \KSub(f))_{*}(m).
\end{array}$$
}

\noindent The functor $\KSub(-)$ preserves the relevant
structure. This requires the following auxiliary result.

\begin{lemma}
\label{lem:kernelschangeofbase}
In a dagger kernel category, for any kernel $k\colon K \rightarrowtail
X$ in $\KSub(X)$, there is an order isomorphism $\KSub(K) \cong
\downset k \subseteq \KSub(X)$.
\end{lemma}

\begin{proof}
The direction $\KSub(K) \to \downset k$ of the desired bijection is
given by $m \mapsto k \after m$. This is well-defined since kernels
are closed under composition.  The other direction $\downset k \to
\KSub(K)$ is $n \mapsto \varphi = k^{\dag} \after n$, where $n = k
\after \varphi$. One easily checks that these maps are each other's
inverse, and preserve the order.  \QED
\end{proof}




\begin{theorem}
\label{KSubPreservationThm}
Let $\cat{D}$ be a dagger kernel category. The functor
$\KSub \colon \cat{D} \to \Cat{OMLatGal}$,
\begin{enumerate}
\item[(a)] is a map of dagger kernel categories;

\item[(b)] preserves (finite) biproducts, in case they exist in $\cat{D}$.
\end{enumerate}
\end{theorem}

\begin{proof}
Preservation of daggers follows because $f^{-1}$ and $\exists_f$ are
inter-expressible, see Section~\ref{DagKerSec}
and~\cite[Proposition~4.3]{HeunenJ09a}:
$$\KSub(f^{\dag})_{*}(n)
=
\big(\exists_{f^{\dag}}(n))^{\perp}
=
f^{-1}(n^{\perp})
=
\KSub(f)^{*}(n)
=
\big(\KSub(f)^{\dag}\big)_{*}(n).$$

Preservation of the zero object is easy: $\KSub(0) = \{0\}=\nul$.

Next, let $k\colon K \to X$ be the kernel of a morphism $f \colon X
\to Y$ in $\cat{D}$. We recall
from~\cite[Corollary~2.5~(ii)]{HeunenJ09a} that this kernel $k$ can be
described as inverse image $k = f^{-1}(0) = f^{-1}(1^{\perp}) =
\KSub(f)^{*}(1)$. Hence by Lemmas~\ref{lem:kernelschangeofbase}
and~\ref{DownsetLem}, we have the isomorphism on the left in:
$$\xymatrix@R-2pc@C+2pc{
\KSub(K)\ar@{ |>->}[dr]^-{\KSub(k)}\ar[dd]_{k\after -}^{\cong} \\
& \KSub(X)\ar[r]^-{\KSub(f)} & \KSub(Y) \\
\downset k\ar@{ |>->}[ur]_{k}
}$$

\noindent It yields a commuting triangle since for $n\in\KSub(K)$,
$$\KSub(k)_{*}(n)
=
\exists_k(n)^{\perp} 
=
(k\after n)^{\perp} 
=
k_{*}(k\after n).$$

\noindent Similarly for $m\in\KSub(X)$, 
$$k\after \KSub(k)^{*}(m)
=
k\after k^{-1}(m^\perp)
= 
k \conjun m^{\perp}
= 
k^{*}(m).$$

For~(b), it suffices to prove that $\KSub(\kappa_i) \colon \KSub(X_i)
\to \KSub(X_1 \oplus X_2)$ is a coproduct in $\Cat{OMLatGal}$.  Let
morphisms $f_i \colon \KSub(X_i) \to Y$ in $\Cat{OMLatGal}$ be
given. Define a cotuple $\cotuple{f_1}{f_2} \colon \KSub(X_1 \oplus
X_2) \to Y$ by
$$\begin{array}{rcl}
\cotuple{f_1}{f_2}_*(m) 
& = &
(f_{1})_{*}(\exists_{\pi_1}(m)) \conjun (f_{2})_{*}(\exists_{\pi_2}(m)), \\
\cotuple{f_1}{f_2}^{*}(y) 
& = &
\pi_{1}^{-1}((f_{1})^{*}(y)) \conjun \pi_{2}^{-1}((f_{2})^{*}(y))
\end{array}$$

\noindent Indeed $\cotuple{f_1}{f_2}^{*} \dashv \cotuple{f_1}{f_2}_{*}$: 
$$\begin{prooftree}
\begin{prooftree}
y \;\leq\; (f_{1})_{*}(\exists_{\pi_1}(m)) \conjun (f_{2})_{*}(\exists_{\pi_2}(m))
   \rlap{$\;=\cotuple{f_1}{f_2}_{*}(m)$}
\Justifies
\begin{prooftree}
\begin{prooftree}
y \;\leq\; (f_{i})_{*}(\exists_{\pi_i}(m))
\Justifies
\exists_{\pi_i}(m) \;\leq\; (f_{i})^{*}(y)
\end{prooftree}
\Justifies
m \;\leq\; \pi_{i}^{-1}((f_{i})^{*}(y))
\end{prooftree}
\end{prooftree}
\Justifies
m \;\leq\; \pi_{1}^{-1}((f_{1})^{*}(y)) \conjun \pi_{2}^{-1}((f_{2})^{*}(y))
   \rlap{$\;= \cotuple{f_1}{f_2}^*(y)$}
\end{prooftree}$$

\noindent This morphism $\cotuple{f_1}{f_2} \colon
  \KSub(X_1 \oplus X_2) \to Y$ of $\Cat{OMLatGal}$ satisfies:
$$\begin{array}{rcl}
(\cotuple{f_1}{f_2}) \after \KSub(\kappa_{1}))_{*}(m)
& = &
\cotuple{f_1}{f_2}_{*}( \KSub(\kappa_{1})_{*}(m)^{\perp}) \\
& = &
\cotuple{f_1}{f_2}_{*}(\exists_{\kappa_1}(m)) \\
& = &
(f_{1})_{*}(\exists_{\pi_1}\exists_{\kappa_1}(m)) \conjun
      (f_{2})_{*}(\exists_{\pi_2}\exists_{\kappa_1}(m)) \\
& = &
(f_{1})_{*}(\exists_{\pi_{1}\after\kappa_1}(m)) \conjun
      (f_{2})_{*}(\exists_{\pi_{2}\after\kappa_1}(m)) \\
& = &
(f_{1})_{*}(\exists_{\idmap}(m)) \conjun
      (f_{2})_{*}(\exists_{0}(m)) \\
& = &
(f_{1})_{*}(m) \conjun (f_{2})_{*}(0) \\
& = &
(f_{1})_{*}(m) \conjun 1 \\
& = &
(f_{1})_{*}(m).
\end{array}$$

\auxproof{
\noindent In the other direction, using that meets $\conjun,1$ are
preserved under pullback,
$$\begin{array}{rcl}
(\cotuple{f_1}{f_2} \after \KSub(\kappa_{1}))^{*}(y)
& = &
\kappa_{1}^{-1}(\cotuple{f_1}{f_2}(y)) \\
& = &
\kappa_{1}^{-1}(\pi_{1}^{-1}((f_{1})^{*}(y))) \conjun 
   \kappa_{1}^{-1}(\pi_{2}^{-1}((f_{2})^{*}(y))) \\
& = &
(\pi_{1}\after \kappa_{1})^{-1}((f_{1})^{*}(y)) \conjun 
   (\pi_{2}\after \kappa_{1})^{-1}((f_{2})^{*}(y)) \\
& = &
\idmap^{-1}((f_{1})^{*}(y)) \conjun 
   0^{-1}((f_{2})^{*}(y)) \\
& = &
(f_{1})^{*}(y) \conjun 1 \\
& = &
(f_{1})^{*}(y).
\end{array}$$
}

\noindent Towards uniqueness, assume $g\colon \KSub(X_{1}\oplus X_{2})
\rightarrow Y$ in \Cat{OMLatGal} also satisfies $g \after
\KSub(\kappa_{i}) = f_{i}$. Then:
$$\begin{array}{rcl}
\cotuple{f_1}{f_2}_{*}(m)
& = &
(f_{1})_{*}(\exists_{\pi_1}(m)) \conjun (f_{2})_{*}(\exists_{\pi_2}(m)) \\
& = &
g_{*}(\exists_{\kappa_{1}}\exists_{\pi_1}(m)) \conjun 
   g_{*}(\exists_{\kappa_2}\exists_{\pi_2}(m)) \\
& = &
g_{*}(\exists_{\kappa_{1}\after\pi_1}(m) \disjun 
   \exists_{\kappa_{2}\after\pi_2}(m)) \\
& = &
g_{*}(m).
\end{array}$$

\noindent This last step needs justification. By
Theorem~\ref{OMLatGalDagKerCatThm}, $m\in\KSub(X_{1}\oplus X_{2})$ can
be written as $m = (\downset(x_{1},x_{2})\rightarrow X_{1}\oplus
X_{2})$ for certain $x_{i}\in X_{i}$. Then, by
Lemma~\ref{InvDirImLem},
$$\begin{array}{rcl}
\lefteqn{\exists_{\kappa_{1}\after\pi_1}(m) \disjun 
   \exists_{\kappa_{2}\after\pi_2}(m)} \\
& = &
\big(\downset((\kappa_{1}\after\pi_{1})_{*}(x_{1},x_{2})^{\perp}
   \disjun (\kappa_{2}\after\pi_{2})_{*}(x_{1},x_{2})^{\perp}) 
   \longrightarrow X_{1}\oplus X_{2}\big) \\
& = &
\big(\downset(x_{1},x_{2})\longrightarrow X_{1}\oplus X_{2}\big) \\
& = &
m,
\end{array}$$

\noindent since:
$$\begin{array}[b]{rcl}
(\kappa_{1}\after\pi_{1})_{*}(x_{1},x_{2})^{\perp}
& = &
(\kappa_{1})_{*}((\pi_{1})_{*}(x_{1},x_{2})^{\perp})^{\perp} \\
& = &
(\kappa_{1})_{*}(x_{1}^{\perp\perp})^{\perp} \\
& = &
(x_{1}^{\perp}, 1)^{\perp} \\
& = &
(x_{1}, 0).
\end{array}\eqno{\QEDbox}$$
\end{proof}

At this stage we conclude that these $\KSub$ functors yield a
well-behaved translation of a dagger kernel category into a collection
of orthomodular lattices, indexed by the objects of the category. For
the special case $\Cat{D} = \Cat{OMLatGal}$, the functor $\KSub\colon
\Cat{D} \rightarrow \Cat{OMLatGal}$ is the identity, up to
isomorphism, by Proposition~\ref{DownsetIsKerProp}. A translation in
the other direction, from orthomodular lattices to dagger kernel
categories will be postponed until after the next section, after we
have seen the translation from Foulis semigroups to orthomodular
lattices.

In the remainder of this section we shall briefly consider two
special subcategories of \Cat{OMLatGal}, namely with Boolean
and with complete orthomodular lattices.

\subsection{The Boolean case}\label{BooleanSubsec}

Let $\Cat{BoolGal} \hookrightarrow \Cat{OMLatGal}$ be the full
subcategory of Boolean algebras with (antitone) Galois connections
between them. We recall that a Boolean algebra can be described as an
orthomodular lattice that is distributive.

The main (and only) result of this subsection is simple.

\begin{proposition}
\label{BoolGalStructProp}
The category \Cat{BoolGal} inherits dagger kernels and biproducts
from \Cat{OMLatGal}. Moreover, as a dagger kernel category it is
Boolean.
\end{proposition}

\begin{proof}
An arbitrary map $f\colon X\rightarrow Y$ in \Cat{BoolGal} has a
kernel $\downset f^{*}(1) \rightarrow X$ as in
Theorem~\ref{OMLatGalDagKerCatThm} for orthomodular lattices because
the downset $\downset f^{*}(1)$ is a Boolean algebra. Similarly, the
biproducts from Proposition~\ref{BiprodProp} also exist in
\Cat{BoolGal} because $X_{1}\oplus X_{2}$ is a Boolean algebra if
$X_{1}$ and $X_{2}$ are Boolean algebras.

For each $X\in\Cat{BoolGal}$ one has $\KSub(X) \cong X$ so that
$\KSub(X)$ is a Boolean algebra. Hence \Cat{BoolGal} is a Boolean
dagger kernel category by~\cite[Theorem~6.2]{HeunenJ09a}. \QED 

\auxproof{
Old, direct proof.

We check the property $m\conjun n = 0 \Rightarrow m^{\dag} \after n = 0$,
for kernels $m,n$, that defines Booleanness for dagger kernel
categories, see~\cite{HeunenJ09a}. Now suppose we have kernels
$\downset a\rightarrow X$ and $\downset b\rightarrow X$ in
\Cat{BoolGal} with
$$\begin{array}{rcccl}
\big(\downset 0 \rightarrow X\big)
& = &
\big(\downset a\rightarrow X\big) \conjun 
   \big(\downset b\rightarrow X\big)
& = &
\big(\downset(a\conjun b)\rightarrow X\big),
\end{array}$$

\noindent see Proposition~\ref{DownsetIsKerProp}. 
Then $a\conjun b = 0$ in $X$, and thus 
$$\begin{array}{rcl}
a
\hspace*{\arraycolsep} = \hspace*{\arraycolsep}
a\conjun 1
\hspace*{\arraycolsep} = \hspace*{\arraycolsep}
a\conjun (a\conjun b)^{\perp}
& = &
a\conjun (a^{\perp} \disjun b^{\perp}) \\
& = &
(a\conjun a^{\perp}) \disjun (a\conjun b^{\perp}) 
   \quad\mbox{by distributivity} \\
& = &
0 \disjun (a \conjun b^{\perp} \\
& = &
a\conjun b^{\perp}.
\end{array}$$

\noindent Hence $a\leq b^{\perp}$, so that the corresponding kernels
satisfy, for $v\in \downset b$,
$$\begin{array}[b]{rcl}
(a^{\dag}\after b)_{*}(v)
\hspace*{\arraycolsep} = \hspace*{\arraycolsep}
a^{*}(b_{*}(v)^{\perp})
& = &
a\conjun v^{\perp} \\
& = &
a \qquad\mbox{since $v\leq b$, so $a \leq b^{\perp} \leq v^{\perp}$} \\
& = &
1_{\downset a} \\
& = &
\big(0\colon \!\downset a\rightarrow \downset a\big)_{*}(v).
\end{array}\eqno{\QEDbox}$$
}
\end{proof}

Boolean algebras thus give rise to (Boolean) dagger kernel categories
on two different levels: the ``large'' category \Cat{BoolGal} of all
Boolean algebras is a dagger kernel category, but also each individual
Boolean algebra can be turned into a ``small'' dagger kernel category,
see~\cite[Proposition~3.5]{HeunenJ09a}.

\subsection{Complete orthomodular lattices}\label{CompleteSubsec}

We shall write $\Cat{OMSupGal}\hookrightarrow \Cat{OMLatGal}$ for the
full subcategory of orthomodular lattices that are complete,
\textit{i.e.}~that have joins $\bigvee U$ (and thus also meets
$\bigwedge U$) of all subsets $U$ (and not just the finite ones).
Notice that the functor $\KSub\colon \Cat{D} \rightarrow
\Cat{OMLatGal}$ from Theorem~\ref{KSubPreservationThm} is actually a
functor $\KSub\colon \Cat{D} \rightarrow \Cat{OMSupGal}$ for \Cat{D} =
\Rel, \PInj, \Hilb.

A morphism $f\colon X\rightarrow Y$ in \Cat{OMSupGal} is completely
determined by either $f_{*}\colon X\op \rightarrow Y$ preserving all
meets, or by $f^{*}\colon Y\rightarrow X\op$ preserving all
joins. This forms the basis for the next result.

\begin{proposition}
\label{FreeOMLatGalProp}
The forgetful functor $U\colon \Cat{OMSupGal} \to \Sets$ given by
$X\mapsto X$ on objects and $f\mapsto f_{*} \after \bot$ on morphisms
has a left adjoint $F$ given by $F(A) = \powerset{A}$, with
$F(g)_{*}(U \subseteq A) = \neg\coprod_{g}(U) = \neg\set{g(a)}{a\in
  U}$ and $F(g)^{*}(V\subseteq B) = g^{-1}(\neg V) =
\set{a}{g(a)\not\in V}$, for $g\colon A\rightarrow B$ in $\Sets$.
\end{proposition}

\begin{proof}
For $A\in\Sets$ and $X\in\Cat{OMSupGal}$ there is a bijective
correspondence:
$$\begin{bijectivecorrespondence}
\correspondence[in \Cat{OMSupGal}]{\xymatrix{\powerset{A}\ar[r]^-{f} & X}}
\correspondence[in \Sets]{\xymatrix{A\ar[r]_-{g} & X}}
\end{bijectivecorrespondence}$$

\noindent given by $\overline{f}(a) = f_{*}(\{a\})^{\perp}$ and
$\overline{g}_{*}(U) = \bigwedge_{a\in U}g(a)^{\perp}$ with
$\overline{g}^{*}(x) = \setin{a}{A}{g(a) \leq x^{\perp}}$. Then:
$$\begin{array}{rcl}
x\leq \overline{g}_{*}(U) = \bigwedge_{a\in U}g(a)^{\perp}
& \Longleftrightarrow &
\allin{a}{U}{x\leq g(a)^{\perp}} \\
& \Longleftrightarrow &
\allin{a}{U}{g(a) \leq x^{\perp}} \\
& \Longleftrightarrow &
U \subseteq \setin{a}{A}{g(a) \leq x^{\perp}} = \overline{g}^{*}(x).
\end{array}$$

\noindent Further,
$$\begin{array}[b]{rcl}
\overline{\overline{g}}(a)
& = & \textstyle
\overline{g}_{*}(\{a\})^{\perp} 
=
\big(\bigwedge_{b\in\{a\}}g(b)^{\perp}\big)^{\perp} 
=
g(a)^{\perp\perp}
=
g(a). \\
\overline{\overline{f}}_{*}(U)
& = & \textstyle
\bigwedge_{a\in U} \overline{f}(a)^{\perp} 
=
\bigwedge_{a\in U} f_{*}(\{a\})^{\perp\perp}
=
f_{*}(\bigcup_{a\in U}\{a\}) 
=
f_{*}(U) \\
\overline{\overline{f}}^{*}(x)
& = & \textstyle
\set{a}{\overline{f}(a) \leq x^{\perp}} 
=
\set{a}{f_{*}(\{a\})^{\perp} \leq x^{\perp}} 
=
\set{a}{x \leq f_{*}(\{a\})} \\
& = &
\set{a}{\{a\}\subseteq f^{*}(x)}
=
f^{*}(x).
\end{array}\eqno{\QEDbox}$$

\auxproof{
We check naturality:
$$\begin{array}{rcl}
\overline{h \after f}(a)
& = &
(h\after f)_{*}(\{a\})^{\perp} \\
& = &
h_{*}(f_{*}(\{a\})^{\perp})^{\perp} \\
& = &
h_{*}(\overline{f}(a))^{\perp} \\
& = &
(h_{*} \after \perp \after \overline{f})(a) \\
\overline{(g \after k)}_{*}(U)
& = &
\bigcup_{a\in U}(g\after k)(a)^{\perp} \\
& = &
\bigcup_{a\in U}g(k(a))^{\perp} \\
& = &
\bigcup_{b\in\coprod_{k}(U)}g(b)^{\perp} \\
& = &
\overline{g}_{*}(\coprod_{k}U) \\
& = &
\overline{g}_{*}(\neg\neg\coprod_{k}(U)) \\
& = &
\overline{g}_{*}(\neg F(k)_{*}(U)) \\
& = &
(\overline{g} \after F(k))_{*}(U) \\
\overline{(g \after k)}^{*}(x)
& = &
\set{b}{g(k(b))\leq x^{\perp}} \\
& = &
\set{b}{k(b)\in\set{a}{g(a) \leq x^{\perp}}} \\
& = &
k^{-1}(\overline{g}^{*}(x)) \\
& = &
k^{-1}(\neg\neg\overline{g}^{*}(x)) \\
& = &
F(k)^{*}(\neg\overline{g}^{*}(x)) \\
& = &
(\overline{g} \after F(k))^{*}(x).
\end{array}$$
}
\end{proof}

The left adjoint $F$ of this adjunction between \cat{OMSupGal} and
\Sets factors via the graph functor ${\cal G}\colon \Sets \rightarrow
\Rel$, as in:
$$\xymatrix@R-2pc@C+1pc{
& \Rel\ar@/^1ex/[dr]^(0.4){\KSub} \\
\Sets\ar@/^1ex/[ur]^-{\cal G} & \quad\bot & \Cat{OMSupGal}\ar@/^2ex/[ll]^{U}
}$$

It is not hard to see that the kernels from
Theorem~\ref{OMLatGalDagKerCatThm} and biproducts $\oplus$ from
Proposition~\ref{BiprodProp} also exist in \Cat{OMSupGal}.  For
instance, the join of a subset $U\subseteq X\oplus Y$ is given as pair
of joins:
$$\begin{array}{rcl}
\bigvee U
& = &
(\bigvee\set{x}{\ex{y}{(x,y)\in U}}, \bigvee\set{y}{\ex{x}{(x,y)\in U}}).
\end{array}$$

\noindent Hence \Cat{OMSupGal} is also a dagger kernel category
with dagger biproducts.


\section{Foulis semigroups and dagger kernel categories}\label{FoulisDagKerSec}

In this section we shall relate dagger kernel categories and Foulis
semigroups. Without a definition yet, we first illustrate that these
Foulis semigroups arise quite naturally in the context of kernel
dagger categories.

In every category \Cat{D} the homset $\EndoHom{X} = \Cat{D}(X, X)$ of
endomaps $f\colon X\rightarrow X$ is a monoid (or semigroup with
unit), with obvious composition operation $\after$ and identity map
$\idmap[X]$ as unit element. If \Cat{D} is a dagger category, there is
automatically an involution $(-)^\dag$ on this monoid. If it is
moreover a dagger kernel category, every endomap $s\in\EndoHom{X}$
yields a self-adjoint idempotent, namely the effect of its kernel:
\begin{equation}
\label{DagKerSaiEqn}
\begin{array}{rcccl}
\sai{s}
& \smash{\stackrel{\textrm{def}}{=}} &
\effect{\ker(s)}
& = &
\ker(s) \after \ker(s)^{\dag} \;\colon\; X\longrightarrow X
\end{array}
\end{equation}

\noindent with the special property that for $t\in\EndoHom{X}$,
$$\begin{array}{rcl}
s\after t = 0
& \Longleftrightarrow &
\exin{r}{\EndoHom{X}}{t = \sai{s} \after r}.
\end{array}$$

\noindent Indeed, if $t=\sai{s}\after r$, then:
$$s \after t
=
s \after \ker(s) \after \ker(s)^{\dag} \after r
=
0 \after \ker(s)^{\dag} \after r 
=
0.$$

\noindent Conversely, if $s\after t = 0$, then there is a map
$f$ in \Cat{D} with $\ker(s) \after f = t$. Hence $t$ satisfies:
$$\sai{s} \after t
=
\ker(s) \after \ker(s)^{\dag} \after \ker(s) \after f
=
\ker(s) \after f
=
t.$$

This structure of an involutive monoid $\langle\EndoHom{X}, \after,
\idmap, \dag\rangle$ with such an operation $\sai{-}\colon \EndoHom{X}
\rightarrow \EndoHom{X}$ has been introduced in the 1960s by
Foulis~\cite{Foulis60,Foulis62,Foulis63} and has since then been
studied under the name `Baer *-semigroup' or `Foulis semigroup',
see~\cite[Chapter~5, \S\S18]{Kalmbach83} for a brief overview.

\begin{definition}
\label{FoulisDef}
A Foulis semigroup consists of a monoid (semigroup with unit) $(S,
\cdot, 1)$ together with two endomaps $(-)^{\dag} \colon S\rightarrow S$
and $\sai{-} \colon S\rightarrow S$ satisfying:
\begin{enumerate}
\item $1^{\dag} = 1$ and $(s\cdot t)^{\dag} = t^{\dag}\cdot s^{\dag}$
  and $s^{\dag\dag} = s$, making $S$ an involutive monoid;

\item $\sai{s}$ is a self-adjoint idempotent, \textit{i.e.}~satisfies
  $\sai{s} \cdot \sai{s} = \sai{s} = \sai{s}^{\dag}$;

\item $0 \;\smash{\stackrel{\textrm{def}}{=}}\; \sai{1}$ is a zero
element: $0 \cdot s = 0 = s\cdot 0$;

\item $s\cdot x = 0$ iff $\ex{y}{x = \sai{s}\cdot y}$.

\end{enumerate}

\noindent Or, equivalently (see~\cite[Chapter~5, \S\S18,
Lemma~1]{Kalmbach83}),
\begin{enumerate}
\item[4$'\!$.] $\sai{0} = 1$ and $s \cdot \sai{s} = 0$ and 
$t = \sai{\sai{t^{\dag} \cdot s^{\dag}}\cdot s}\cdot t$.
\end{enumerate}

We form a category \Cat{Fsg} of such Foulis semigroups with monoid
homomorphisms that commute with $\dag$ and $\sai{-}$ as morphisms.
\end{definition}

\auxproof{
We show the equivalence of $(4)$ and $(4')$.

Assume $(4)$. Since $0\cdot 1 = 0$ we have $1 = \sai{0}\cdot y$ for
some $y\in S$. Hence $\sai{0} = \sai{0}\cdot 1 = \sai{0} \cdot \sai{0} \cdot y
= \sai{0} \cdot y = 1$. 

Next, $\sai{s} \cdot 1 = \sai{s}$, so that $\ex{y}{\sai{s} = \sai{s}\cdot y}$
and thus $s\cdot \sai{s} = 0$.

Finally, in order to prove $t = \sai{\sai{t^{\dag} \cdot s^{\dag}}\cdot
s}\cdot t$ we first note:
$$\begin{array}{rcl}
\sai{t^{\dag}\cdot s^{\dag}} \cdot s\cdot t
& = &
\big(\sai{(s\cdot t)^{\dag}} \cdot s\cdot t\big)^{\dag\dag} \\
& = &
\big((s\cdot t)^{\dag} \cdot \sai{(s\cdot t)^{\dag}}^{\dag}\big)^{\dag} \\
& = &
\big((s\cdot t)^{\dag} \cdot \sai{(s\cdot t)^{\dag}}\big)^{\dag} \\
& = &
0^{\dag} \\
& = &
0.
\end{array}$$

\noindent Hence $t = \sai{\sai{t^{\dag}\cdot s^{\dag}} \cdot s}\cdot y$,
for some $y\in S$, and thus:
$$\sai{\sai{t^{\dag}\cdot s^{\dag}} \cdot s}\cdot t
=
\sai{\sai{t^{\dag}\cdot s^{\dag}} \cdot s}\cdot 
   \sai{\sai{t^{\dag}\cdot s^{\dag}} \cdot s}\cdot y
=
\sai{\sai{t^{\dag}\cdot s^{\dag}} \cdot s}\cdot y
=
t.$$

Conversely, assume $(4)$. If $s\cdot x = 0$, then 
$$\begin{array}{rcl}
x
& = &
\sai{\sai{x^{\dag} \cdot s^{\dag}}\cdot s}\cdot x \\
&^ = &
\sai{\sai{(s\cdot x)^{\dag}}\cdot s}\cdot x \\
& = &
\sai{\sai{0^{\dag}}\cdot s}\cdot x \\
& = &
\sai{\sai{0}\cdot s}\cdot x \\
& = &
\sai{1\cdot s}\cdot x \\
& = &
\sai{s}\cdot x.
\end{array}$$

\noindent And if $x = \sai{s}\cdot y$, then $s\cdot x = s\cdot \sai{s} \cdot y
= 0\cdot y = 0$.
}

The constructions before this definition show that for each object
$X\in\Cat{D}$ of a (locally small) dagger kernel category \Cat{D}, the
homset $\EndoHom{X} = \Cat{D}(X, X)$ of endomaps on $X$ is a Foulis
semigroup. Functoriality of this construction is problematic: for an
arbitrary map $f\colon X\rightarrow Y$ in \Cat{D} there is a mapping
$\EndoHom{X} \rightarrow \EndoHom{Y}$, namely $s\mapsto f \after s
\after f^{\dag} \colon Y\rightarrow X\rightarrow X\rightarrow
Y$, but it does not preserve the structure of Foulis semigroups,
and thus only gives rise to presheaf.

\begin{proposition}
\label{DagKer2FoulisProp}
For a dagger kernel category \Cat{D}, each endo homset $\EndoHom{X}$,
for $X\in\Cat{D}$, is a Foulis semigroup. The mapping $X\mapsto
\EndoHom{X}$ yields a presheaf $\Cat{D} \rightarrow \Sets$. \QED
\end{proposition}

The lack of functoriality in this construction is problematic. One
possible way to address it is via another notion of morphism between
Foulis semigroups, like Galois connections between orthomodular
lattices in the category \Cat{OMLatGal}. We shall not go deeper into
this issue. Also the possible sheaf-theoretic aspects involved in this
situation (see also~\cite{GravesS73}) form a topic on its own that is
not pursued here.  We briefly consider some examples.

For the dagger kernel category \Hilb of Hilbert spaces, the set ${\cal
  B}(H)$ of (bounded/continuous linear) endomaps on a Hilbert space
$H$ forms a Foulis semigroup---but of course also a $C^*$-algebra. The
associated (Foulis) map $\sai{-}\colon {\cal B}(H)\rightarrow
{\cal B}(H)$ maps $s\colon H\rightarrow H$ to $\sai{s}\colon H\rightarrow
H$ given by $\sai{s}(x) = k(k^{\dag}(x))$, where $k$ is the kernel map
$\set{x}{s(x)=0} \hookrightarrow H$.


For the category \Rel of sets and relations the endomaps on a set $X$
are the relations $R\subseteq X\times X$ on $X$. The associated
$\sai{R} \subseteq X\times X$ is $\set{(x,x)}{\neg\ex{y}{R(x,y)}}$.

An interesting situation arises when we apply the previous proposition
to the dagger kernel category \Cat{OMLatGal} of orthomodular lattices
(with Galois connections between them). One gets that for each
orthomodular lattice $X$ the endo-homset $\EndoHom{X} =
\Cat{OMLatGal}(X, X)$ forms a Foulis semigroup. This construction is
more than 40 years old, see~\cite{Foulis60} or
\textit{e.g.}~\cite[Chapter~II, Section~19]{BlythJ72}
or~\cite[Chapter~5, \S\S18]{Kalmbach83}, where it is described in
terms of Galois connections.  In the present setting it comes for
free, from the structure of the category \Cat{OMLatGal}. Hence we
present it as a corollary, in particular of
Proposition~\ref{DagKer2FoulisProp} and
Theorem~\ref{OMLatGalDagKerCatThm}.

\begin{corollary}
\label{OMLat2FoulisCor}
For each orthomodular lattice $X$ the set of (Galois) endomaps
$\EndoHom{X} = \Cat{OMLatGal}(X,X)$ is a Foulis semigroup with
composition as monoid, dagger $(-)^{\dag}$ as involution, and
self-adjoint idempotent $\sai{s}\colon X\rightarrow X$, for $s\colon
X\rightarrow X$, defined as in~(\ref{DagKerSaiEqn}). Equivalently,
$\sai{s}$ can be described via the Sasaki hook $\sasaki$ or and-then
operator $\andthen$:
$$\sai{s}_{*}(x) 
= 
\sai{s}^{*}(x) 
=
s^{*}(1) \sasaki x^{\perp}
=
s^{*}(1)^{\perp} \disjun (s^{*}(1) \conjun x^{\perp})
=
(x\andthen s^{*}(1))^{\perp}.$$
\end{corollary}

\begin{proof}
We recall from~(\ref{DagKerSaiEqn}) that the operation $\sai{-}$ on
endomaps $s\colon X\rightarrow X$ is defined as $\sai{s} = \ker(s)
\after \ker(s)^{\dag} \colon X\rightarrow X$. In \Cat{OMLatGal} one
has $\ker(s) = s^{*}(1)$---see Proposition~\ref{DownsetIsKerProp}---so
that:
$$\begin{array}[b]{rcl}
\sai{s}_{*}(x)
& = &
(\ker(s) \after \ker(s)^{\dag})_{*}(x) \\
& = &
(\ker(s)_{*}(\ker(s)^{*}(x)^{\perp}) \\
& = &
s^{*}(1)_{*}\big(s^{*}(1) \conjun s^{*}(1)^{*}(x)^{\perp}\big) \\
& = &
\big(s^{*}(1) \conjun (s^{*}(1) \conjun x^{\perp})^{\perp}\big)^{\perp}
   \qquad\mbox{by Lemma~\ref{DownsetLem}} \\
&= &
s^{*}(1)^{\perp} \disjun (s^{*}(1) \conjun x^{\perp}).
\end{array}\eqno{\QEDbox}$$
\end{proof}

\auxproof{
Old, direct proof.

Obviously $(\idmap[X])^{\dag} = \idmap[X]$ and $(t\after s)^{\dag} =
s^{\dag} \after t^{\dag}$ and $s^{\dag\dag} = s$, for endomaps
$s,t\colon X\rightarrow X$ in \Cat{OMLatGal}. Also, $\sai{s}$ as defined
above is a morphism in \Cat{OMLatGal}, since for $x,y\in X$,
$$\begin{prooftree}
y \;\leq\; \sai{s}_{*}(x) \rlap{$\;= s^{*}(1) \sasaki x^{\perp}$}
\Justifies
\begin{prooftree}
y \andthen s^{*}(1) \;\leq\; x^{\perp}
\Justifies
x \;\leq\; (y\andthen s^{*}(1))^{\perp} \rlap{$\;= \sai{s}^{*}(y)$}
\end{prooftree}
\end{prooftree}$$

\noindent These $\sai{s}$'s are indeed projections: by construction,
$\sai{s}^{\dag} = \sai{s}$ and:
$$\begin{array}{rcl}
(\sai{s}_{*} \after \sai{s}_{*})(x)
& = &
\sai{s}_{*}\big(\sai{s}_{*}(x)^{\perp}\big) \\
& = &
s^{*}(1) \sasaki \sai{s}_{*}(x) \\
& = &
s^{*}(1) \sasaki (s^{*}(1) \sasaki x^{\perp}) \\
& = &
s^{*}(1) \sasaki x^{\perp} \\
& = &
\sai{s}_{*}(x).
\end{array}$$

\noindent Further,
$$\sai{\idmap[X]}_{*}(x)
=
\idmap[X](1) \sasaki x^{\perp}
=
1^{\perp} \sasaki x^{\perp}
=
0 \sasaki x^{\perp}
=
1.$$

\noindent Hence $\sai{\idmap[X]}$ is the zero map $0\colon X\rightarrow
X$, satisfying $0 \after s = 0 = s \after 0$ and $0^{\dag} = 0$, as we
have seen before. Also, this zero map satisfies $\sai{0} = \idmap[X]$
since: $\sai{0}_{*}(x) = 1 \sasaki x^{\perp} = x^{\perp} =
(\idmap[X])_{*}(x)$. From $x \andthen s^{*}(1) \leq s^{*}(1)$ we get,
via the Galois connection, $1 \leq s_{*}(x \andthen s^{*}(1)) =
s_{*}(\sai{s}_{*}(x)^{\perp}) = (s \after \sai{s})_{*}(x)$. Hence $s \after
\sai{s} = 0$. We immediately use this in:
$$\sai{t^{\dag} \after s^{\dag}} \after s \after t
=
\sai{(s\after t)^{\dag}}^{\dag} \after (s\after t)^{\dag\dag}
=
(s\after t)^{\dag} \after \sai{(s\after t)^{\dag}}
=
0.$$

\noindent If we put $r = \sai{t^{\dag} \after s^{\dag}} \after s$ then we
have just seen that $r\after t = 0$. We have to prove $\sai{r} \after t =
t$. As first step, observe that for any $x\in X$,
$$1
=
0_{*}(x)
=
(r \after t)_{*}(x)
=
r_{*}(t_{*}(x)^{\perp}).$$

\noindent By the adjunction: $t_{*}(x)^{\perp} \leq r^{*}(1)$ and
thus $r^{*}(1)^{\perp} \leq t_{*}(x)$. Now we are done by
orthomodularity:
$$t_{*}(x)
=
r^{*}(1)^{\perp} \disjun (r^{*}(1) \conjun t_{*}(x))
=
\sai{r}_{*}(t_{*}(x)^{\perp})
=
(\sai{r} \after t)_{*}(x).\eqno{\QEDbox}$$
}

\subsection{From Foulis semigroups to dagger kernel categories}

Each involutive monoid $(S, \cdot, 1, \dag)$ forms a dagger category
with one object, and morphisms given by elements of
$S$. Requirement~(4) in Definition~\ref{FoulisDef} says that this
category has ``semi'' kernels, given by $\sai{-}$. Hence it is natural
to apply the Karoubi envelope to obtain proper kernels. It turns out
that this indeed yields a dagger kernel category.

For a Foulis semigroup as in Definition~\ref{FoulisDef}, we thus write
$\dagKaroubi{S}$ for the dagger Karoubi envelope applied to $S$ as
one-object dagger category. Thus $\dagKaroubi{S}$ has self-adjoint
idempotents $s\in S$ as objects, and morphisms $f\colon s\rightarrow
t$ given by elements $f\in S$ with $f\cdot s = f = t\cdot f$.

\begin{theorem}
\label{Foulis2DagKerThm}
This $\dagKaroubi{S}$ is a dagger kernel category. The mapping
$S\mapsto \dagKaroubi{S}$ yields a functor $\Cat{Fsg} \rightarrow
\Cat{DKC}$.
\end{theorem}

\begin{proof}
The zero element $0 = \sai{1} \in S$ is obviously a self-adjoint
idempotent, and thus an object of $\dagKaroubi{S}$. It is a zero
object because for each $s\in\dagKaroubi{S}$ there is precisely one
map $f\colon s\rightarrow 0$, namely $0$, because $f = 0\cdot f = 0$.

For an arbitrary map $f\colon s\rightarrow t$ in $\dagKaroubi{S}$
we claim that there is a dagger kernel of the form:
$$\xymatrix{
s\cdot \sai{f} \ar@{ |>->}[rr]^-{s\cdot \sai{f}} & & s\ar[rr]^-{f} & & t
}$$

\noindent This will be checked in a number of small steps.
\begin{itemize}
\item $f\cdot (s\cdot \sai{f}) = (f\cdot s) \cdot \sai{f} = f\cdot \sai{f} = 0$,
by~$(4')$ in Defintion~\ref{FoulisDef};

\item By the previous point there is an element $y\in S$ with
$s\cdot \sai{f} = \sai{f}\cdot y$. Hence:
$$\sai{f}\cdot s\cdot \sai{f}
=
\sai{f}\cdot \sai{f} \cdot y
=
\sai{f}\cdot y
=
s\cdot \sai{f}$$

\noindent This is equation is very useful. It yields first of all that
$s\cdot \sai{f}$ is idempotent:
$$(s\cdot \sai{f})\cdot (s\cdot \sai{f})
=
s\cdot (\sai{f}\cdot s\cdot \sai{f})
=
s\cdot s\cdot \sai{f}
=
s\cdot \sai{f}.$$

\noindent This element is also self-adjoint:
$$(s\cdot \sai{f})^{\dag}
=
(\sai{f}\cdot s\cdot \sai{f})^{\dag} 
=
\sai{f}^{\dag} \cdot s^{\dag} \cdot \sai{f}^{\dag} 
=
\sai{f}\cdot s\cdot \sai{f} 
=
s\cdot \sai{f}.$$

\noindent Hence $s\cdot \sai{f}\in S$ is a self-adjoint idempotent, and
thus an object of $\dagKaroubi{S}$.

\item $s\cdot \sai{f}\colon (s\cdot \sai{f})\rightarrow s$ is also a dagger mono:
$$\begin{array}{rcl}
(s\cdot \sai{f})^{\dag} \cdot (s\cdot \sai{f})
& = &
\sai{f}^{\dag} \cdot s^{\dag} \cdot s \cdot \sai{f} \\
& = &
\sai{f} \cdot s \cdot s \cdot \sai{f} \\
& = &
\sai{f} \cdot s \cdot \sai{f} \\
& = &
s\cdot \sai{f} \\
& = &
\idmap[s\cdot \sai{f}].
\end{array}$$

\item Finally, if $g\colon r\rightarrow s$ in $\dagKaroubi{S}$
  satisfies $f \after g = f\cdot g = 0$, then there is a $y\in S$
  with $g = \sai{f}\cdot y$. Then:
$$s\cdot \sai{f}\cdot g
=
s\cdot \sai{f} \cdot \sai{f} \cdot y
=
s\cdot \sai{f} \cdot y
=
s\cdot g
= 
g.$$

\noindent Hence $g$ is the mediating map $r\rightarrow (s\cdot \sai{f})$,
since $(s\cdot \sai{f})\cdot g = g$. Uniqueness follows because $s\cdot
\sai{f}$ is a dagger mono.
\end{itemize}

As to functoriality, assume $h\colon S\rightarrow T$ is a morphism of
Foulis semigroups. It yields a functor $H\colon \dagKaroubi{S}
\rightarrow \dagKaroubi{T}$ by $s \mapsto h(s)$ and $f\mapsto h(f)$.
This $H$ preserves all the dagger kernel structure because it
preserves the Foulis semigroup structure. \QED
\end{proof}

By combining this result with Proposition~\ref{DagKer2FoulisProp}
we have a way of producing new Foulis semigroups from old.

\begin{corollary}
\label{FoulisEndoCor}
Each self-adjoint idempotent $s\in S$ in a Foulis semigroup $S$ yields
a Foulis semigroup of endo-maps:
$$\EndoHom{s}
=
\dagKaroubi{S}(s, s)
=
\setin{t}{S}{s\cdot t = t = t\cdot s},$$

\noindent with composition $\cdot$, unit $s$, involution $\dag$ and
$\sai{t}_{s} \,\smash{\stackrel{\textrm{def}}{=}}\, s\cdot
\sai{t}\cdot s$. The special case $s=1$ yields the original semigroup:
$\EndoHom{1} = S$.
\end{corollary}

\begin{proof}
We only check the formulation following~(\ref{DagKerSaiEqn}):
$$\sai{t}_{s}
=
\ker(t) \after \ker(t)^{\dag} 
=
s\cdot \sai{t} \cdot (s\cdot \sai{t})^{\dag}
=
s\cdot \sai{t} \cdot \sai{t} \cdot s
=
s\cdot \sai{t}\cdot s.\eqno{\QEDbox}$$
\end{proof}

The posets of kernel subobjects in a dagger kernel category are
orthomodular lattices. This applies in particular to the category
$\dagKaroubi{S}$ and yields a way to construct orthomodular lattices
out of Foulis semigroups. We first investigate this lattice structure
in more detail, via (isomorphic) subsets of $S$.

\begin{lemma}
\label{FoulisOMKerLem}
Let $S$ be a Foulis semigroup with self-adjoint idempotent $s\in S$,
considered as object $s\in\dagKaroubi{S}$. The subset
$$\begin{array}{rcccl}
K_{s}
& \smash{\stackrel{\textrm{def}}{=}} &
\set{s\cdot \sai{t\cdot s}}{t\in S} 
& \subseteq & 
S,
\end{array}$$

\noindent is an orthomodular lattice with the following structure.
$$\begin{array}{lrcl}
\mbox{Order} & k_{1}\leq k_{2} & \Leftrightarrow & k_{1} = k_{2}\cdot k_{1} \\
\mbox{Top} & 1_{s} & = & s \hspace*{\arraycolsep} = \hspace*{\arraycolsep}
   s\cdot \sai{s\cdot 0} \\
\mbox{Orthocomplement} & k^{\perp} & = & s\cdot \sai{k} \\
\mbox{Meet} & k_{1} \conjun k_{2} & = &  
   \big(k_{1} \cdot \sai{\sai{k_{2}}\cdot k_{1}}^{\perp\perp}.
\end{array}$$

\noindent In fact, $K_{s} \cong \KSub(s)$.
\end{lemma}

\begin{proof}
In fact it suffices to prove the last isomorphism $K_{s} \cong
\KSub(s)$ and use it to translate the orthomodular structure from
$\KSub(s)$ to $K_s$. Instead we proceed in a direct manner and show
that each $K_s$ is an orthomodular lattice in a number of small
consecutive steps, resembling the steps taken in~\cite[Chapter~5,
  \S\S18]{Kalmbach83}. One observation that is used a number of times
is:
$$\begin{array}{rcl}
x\cdot y = 0 
& \Longrightarrow &
y = \sai{x}\cdot y
\end{array}\eqno{(*)}$$

\noindent for arbitrary $x,y\in S$, Indeed, if $x\cdot y = 0$, then by
requirement~(4) in Definition~\ref{FoulisDef} there is a $z$ with
$y = \sai{x}\cdot z$. But then $\sai{x}\cdot y = \sai{x}\cdot \sai{x}\cdot z =
\sai{x}\cdot z = y$.

Let $s\in S$ now be a fixed self-adjoint idempotent. 
\begin{enumerate}
\renewcommand{\theenumi}{(\alph{enumi})}
\item Each $k\in K_s$ is a self-adjoint idempotent, a dagger kernel
$k\colon k\rightarrow s$, and also an idempotent $k\colon s\rightarrow s$
in $\dagKaroubi{S}$.

Indeed, if $k = s\cdot \sai{t\cdot s}$, then $(t\cdot s)\cdot k = t\cdot
s \cdot \sai{t\cdot s} = 0$, so that $k = \sai{t\cdot s}\cdot k$ by $(*)$. 
Hence:
\begin{align*}
k\cdot k 
& =
s\cdot \sai{t\cdot s}\cdot k
=
s\cdot k
=
k \\
k^{\dag}
& = 
(\sai{t\cdot s}\cdot k)^{\dag}
=
(\sai{t\cdot s}\cdot s\cdot \sai{t\cdot s})^{\dag} \\
& =
\sai{t\cdot s}^{\dag} \cdot s^{\dag} \cdot \sai{t\cdot s}^{\dag} 
=
\sai{t\cdot s}\cdot s\cdot \sai{t\cdot s}
=
k \\
k^{\dag} \cdot k
& =
k \cdot k
=
k \\
k\cdot s
& =
k^{\dag} \cdot s^{\dag}
=
(s\cdot k)^{\dag}
=
k^{\dag}
=
k.
\end{align*}

\noindent Also, $k\colon k\rightarrow s$ is the kernel of $t\cdot
s\colon s\rightarrow 1$, using the description of kernels in
$\dagKaroubi{S}$ from the proof of Theorem~\ref{Foulis2DagKerThm}.

\item The set $S$ carries a transitive order $t\leq r$ iff
$r\cdot t = t$. This $\leq$ is a partial order on $K_s$.

Transitivity is obvious: if $t\leq r \leq q$, then $r\cdot t = t$ and
$q\cdot r = r$ so that $q\cdot t = q\cdot r\cdot t = r\cdot t = t$,
showing that $t\leq q$.

Reflexivity $k\leq k$ holds for $k\in K_s$ since we have $k\cdot k =
k$ as shown in~(a). For symmetry assume $k\leq \ell$ and $\ell\leq k$
where $k,\ell\in K_{s}$. Then $\ell\cdot k = k$ and $k\cdot \ell =
\ell$. Hence $k = k^{\dag} = (\ell\cdot k)^{\dag} = k^{\dag} \cdot \ell^{\dag}
= k\cdot \ell = \ell$.

\item For an arbitrary $t\in S$ put $t^{\perp}
  \,\smash{\stackrel{\textrm{def}}{=}}\, s\cdot \sai{t^{\dag}\cdot s} \in
  K_{s}$. Hence from~(a) we get equations $t^{\perp} \cdot t^{\perp} =
  t^{\perp} = (t^{\perp})^{\dag}$ and $s\cdot t^{\perp} = t^{\perp} =
  t^{\perp}\cdot s$ that are useful in calculations.

We will show $t\leq r \Rightarrow r^{\perp} \leq t^{\perp}$
and $k^{\perp\perp} = k$ for $k\in K_{s}$.

Assume $t\leq r$, \textit{i.e.}~$t = r\cdot t$. Then, applying the
equation $y = \sai{\sai{y^{\dag}\cdot x^{\dag}}\cdot x} \cdot y$ from
requirement~$(4')$ in Definition~\ref{FoulisDef} for $y =
\sai{r^{\dag}\cdot s}$ and $x= t^{\dag}\cdot s$ yields:
$$\begin{array}{rcl}
\sai{r^{\dag}\cdot s}
& = &
\sai{\sai{\sai{r^{\dag}\cdot s}^{\dag}\cdot (t^{\dag}\cdot s)^{\dag}}
   \cdot t^{\dag}\cdot s}\cdot \sai{r^{\dag}\cdot s} \\
& = &
\sai{\sai{\sai{(t^{\dag}\cdot s\cdot \sai{r^{\dag}\cdot s})^{\dag}}}
   \cdot t^{\dag}\cdot s}\cdot \sai{r^{\dag}\cdot s} \\
& = &
\sai{\sai{((r\cdot t)^{\dag}\cdot s\cdot \sai{r^{\dag}\cdot s})^{\dag}}
   \cdot t^{\dag}\cdot s}\cdot \sai{r^{\dag}\cdot s} \\
& = &
\sai{\sai{(t^{\dag}\cdot r^{\dag}\cdot s\cdot \sai{r^{\dag}\cdot s})^{\dag}}
   \cdot t^{\dag}\cdot s}\cdot \sai{r^{\dag}\cdot s} \\
& = &
\sai{\sai{(t^{\dag} \cdot 0)^{\dag}} 
   \cdot t^{\dag}\cdot s}\cdot \sai{r^{\dag}\cdot s}
   \qquad\mbox{since $x\cdot \sai{x}=0$} \\
& = &
\sai{\sai{0} \cdot t^{\dag}\cdot s}\cdot \sai{r^{\dag}\cdot s} \\
& = &
\sai{1 \cdot t^{\dag}\cdot s}\cdot \sai{r^{\dag}\cdot s} \\
& = &
\sai{t^{\dag}\cdot s}\cdot \sai{r^{\dag}\cdot s}.
\end{array}$$

\noindent This gives us what we need to show $r^{\perp} \leq t^{\perp}$:
$$\begin{array}{rcll}
t^{\perp}\cdot r^{\perp}
& = &
t^{\perp} \cdot s\cdot \sai{r^{\dag}\cdot s} \\
& = &
t^{\perp} \cdot \sai{r^{\dag}\cdot s} 
   & \mbox{since $t^{\perp}\in K_s$} \\
& = &
s\cdot \sai{t^{\dag}\cdot s}\cdot \sai{r^{\dag}\cdot s} \\
& = &
s\cdot \sai{r^{\dag}\cdot s}
   & \mbox{as we have just seen} \\
& = &
r^{\perp}.
\end{array}$$

Next we notice that 
$$t^{\perp\perp} 
= 
s\cdot \sai{(t^{\perp})^{\dag}\cdot s} 
=
s\cdot \sai{\sai{t^{\dag}\cdot s}^{\dag}\cdot s^{\dag} \cdot s}
=
s\cdot \sai{\sai{t^{\dag}\cdot s}\cdot s}.$$

\noindent Requirement~$(4')$ in Definition~\ref{FoulisDef}, applied to $t$, 
says:
$$s\cdot t 
= 
s\cdot \sai{\sai{t^{\dag}\cdot s^{\dag}}\cdot s}\cdot t
=
s\cdot \sai{\sai{t^{\dag}\cdot s}\cdot s}\cdot t
=
t^{\perp\perp} \cdot t
=
t^{\perp\perp} \cdot s\cdot t.$$

\noindent It says that $s\cdot t \leq t^{\perp\perp}$. In particular,
this means $k\leq k^{\perp\perp}$ for $k\in K_s$. Since $(-)^{\perp}$
reverses the order we get:
$$t^{\perp\perp\perp}
\leq 
(s\cdot t)^{\perp}
=
s\cdot \sai{(s\cdot t)^{\dag}\cdot s}
=
s\cdot \sai{t^{\dag}\cdot s^{\dag} \cdot s}
=
s\cdot \sai{t^{\dag} \cdot s}
=
t^{\perp}.$$

\noindent If we finally apply this to $k\in K_{s}$, say for
$k = s\cdot \sai{t\cdot s} = (t^{\dag})^{\perp}$ we get:
$$k^{\perp\perp}
=
(t^{\dag})^{\perp\perp\perp}
\leq
(t^{\dag})^{\perp}
=
k.$$

\item As motivation for the definition of meet, consider for $k_{1},
  k_{2}\in K_{s}$ their meet as kernels:
$$\begin{array}{rcl}
r
& \smash{\stackrel{\textrm{def}}{=}} &
k_{1}\cdot k_{1}^{-1}(k_{2}) \\
& = &
k_{1}\cdot \ker(\coker(k_{2})\cdot k_{1}) 
   \qquad\mbox{see pullback from Section~\ref{DagKerSec}} \\
& = &
k_{1}\cdot \ker((s\cdot \sai{k_{2}^{\dag}})^{\dag} \cdot k_{1}) \\
& = &
k_{1} \cdot k_{1} \cdot \sai{\sai{k_{2}}\cdot s \cdot k_{1}} \\
& = &
k_{1} \cdot \sai{\sai{k_{2}}\cdot k_{1}}.
\end{array}$$

\noindent We force this $r$ into $K_{s}$ via double negation and hence
define $k_{1}\conjun k_{2} = r^{\perp\perp}$. Showing that it is the
meet of $k_{1}, k_{2}$ requires a bit of work.
\begin{itemize}
\item We have $k_{1}\cdot r = k_{1} \cdot k_{1} \cdot \sai{\sai{k_{2}}\cdot
  k_{1}} = k_{1} \cdot \sai{\sai{k_{2}}\cdot k_{1}} = r$, so that $r \leq
  k_{1}$ and thus also $k_{1} \conjun k_{2} = r^{\perp\perp} \leq
  k_{1}^{\perp\perp} = k_{1}$.

\item We first observe that
$$\sai{k_{2}}\cdot s \cdot s \cdot r
=
\sai{k_{2}}\cdot s \cdot s \cdot k_{1} \cdot \sai{\sai{k_{2}}\cdot k_{1}}
=
\sai{k_2}\cdot k_{1} \cdot \sai{\sai{k_{2}}\cdot k_{1}}
=
0.$$

\noindent Hence by applying $\dag$ we get $(r^{\dag} \cdot s) \cdot
(s\cdot \sai{k_2}) = 0$. Via $(*)$ we obtain $s\cdot \sai{k_2} = 
\sai{r^{\dag}\cdot s} \cdot s \cdot \sai{k_2}$, and thus also
$$k_{2}^{\perp} 
= 
s\cdot k_{2}^{\perp} 
= 
s\cdot s\cdot \sai{k_2} 
= 
s\cdot \sai{r^{\dag}\cdot s} \cdot s \cdot \sai{k_2}
=
r^{\perp}\cdot k_{2}^{\perp}.$$

\noindent This says $k_{2}^{\perp} \leq r^{\perp}$, from
which we get $k_{1} \conjun k_{2} = r^{\perp\perp} \leq k_{2}^{\perp\perp}
= k_{2}$.

\item If also $\ell\in K_{s}$ satisfies $\ell \leq k_{1}$ and $\ell\leq k_{2}$,
\textit{i.e.}~$k_{1}\cdot \ell = \ell = k_{2}\cdot \ell$, then, by
Definition~\ref{FoulisDef}~$(4')$,
$$\begin{array}{rcl}
\sai{k_2}\cdot k_{1}\cdot \ell
\hspace*{\arraycolsep} = \hspace*{\arraycolsep}
\sai{k_2}\cdot k_{2} \cdot \ell
& = &
(\sai{k_2}\cdot k_{2})^{\dag\dag}\cdot\ell \\
& = &
(k_{2}^{\dag} \cdot \sai{k_2}^{\dag})^{\dag}\cdot\ell \\
& = &
(k_{2} \cdot \sai{k_2})^{\dag}\cdot\ell \\
& = &
0^{\dag}\cdot\ell \\
& = &
0.
\end{array}$$

\noindent Hence $\ell = \sai{\sai{k_{2}}\cdot k_{1}} \cdot \ell$
by~$(*)$ and so $\ell = k_{1}\cdot \ell = s\cdot k_{1}\cdot \ell =
s\cdot k_{1} \cdot \sai{\sai{k_{2}}\cdot k_{1}} \cdot \ell = s\cdot r\cdot
\ell$. Thus $\ell\leq s\cdot r \leq r^{\perp\perp} = k_{1} \conjun k_{2}$.
\end{itemize}

\item We get $k^{\perp}\conjun k = 0$, for $k\in K_{s}$, as follows.
  Since $k\cdot s\cdot \sai{k} = k\cdot \sai{k} = 0$ one has $k^{\perp} =
  s\cdot \sai{k} = \sai{k} \cdot s\cdot \sai{k} = \sai{k}\cdot k^{\perp}$
  by~$(*)$. Hence:
$$k^{\perp}\conjun k
=
\big(k^{\perp}\cdot \sai{\sai{k}\cdot k^{\perp}}\big)^{\perp\perp}
=
\big(k^{\perp} \cdot \sai{k^{\perp}}\big)^{\perp\perp} 
=
0^{\perp\perp} 
=
0.$$

\item Finally, orthomodularity holds in $K_{s}$. We assume $k\leq
  \ell$ (\textit{i.e.}~$k = \ell\cdot k$) and $k^{\perp}\conjun
  \ell=0$, for $k,\ell\in K_{s}$, and have to show $\ell\leq k$
  (\textit{i.e.}~$\ell = k\cdot \ell$, and thus $k=\ell$). To start,
  $k = k^{\dag} = (\ell\cdot k)^{\dag} = k^{\dag}\cdot \ell^{\dag} =
  k\cdot \ell$, so that $k\cdot \ell^{\perp} = k\cdot s\cdot \sai{\ell} =
  k\cdot \sai{\ell} = k\cdot \ell \cdot \sai{\ell} = k\cdot 0 =
  0$. Using~$(*)$ yields $\ell^{\perp} = \sai{k}\cdot \ell^{\perp} = \sai{k}
  \cdot s \cdot \sai{\ell}$, and also $\ell^{\perp} =
  (\ell^{\perp})^{\dag} = (\sai{k}\cdot s \cdot \sai{\ell})^{\dag} =
  \sai{\ell}^{\dag} \cdot s^{\dag} \cdot \sai{k}^{\dag} = \sai{\ell} \cdot
  k^{\perp}$. Hence:
$$\begin{array}{rcl}
k^{\perp}\cdot \ell
\hspace*{\arraycolsep} = \hspace*{\arraycolsep}
k^{\perp}\cdot \ell^{\perp\perp}
\hspace*{\arraycolsep} = \hspace*{\arraycolsep}
k^{\perp}\cdot s\cdot \sai{\ell^{\perp}}
& = &
k^{\perp} \cdot \sai{\ell^{\perp}} \\
& = &
k^{\perp} \cdot \sai{\sai{\ell} \cdot k^{\perp}} \\
& \leq &
\big(k^{\perp} \cdot \sai{\sai{\ell} \cdot k^{\perp}}\big)^{\perp\perp} \\
& = &
k^{\perp} \conjun \ell \\
& = &
0.
\end{array}$$

\noindent By~$(*)$ we get $\ell = \sai{k^{\perp}} \cdot \ell$ so that
$\ell = s\cdot \ell = s\cdot \sai{k^{\perp}} \cdot \ell =
k^{\perp\perp}\cdot \ell = k\cdot \ell$, as required to get
$\ell \leq k$.
\end{enumerate}

Finally we need to show $K_{s} \cong \KSub(s)$. As we have seen
in~(a), each $k\in K_{s}$ yields (an equivalence class of) a kernel
$k\colon k\rightarrow s$. Conversely, each kernel $\ker(f) = s\cdot
\sai{f} = s\cdot \sai{f\cdot s}$ of a map $f\colon s\rightarrow t$ in
$\dagKaroubi{S}$---see the proof of
Theorem~\ref{Foulis2DagKerThm}---is an element of $K_s$. This yields
an order isomorphism: if $k_{1}\leq k_{2}$ for $k_{1}, k_{2} \in
K_{s}$, then $k_{1} = k_{2} \cdot k_{1}$ so that we get a commuting
triangle:
$$\xymatrix@R-2pc{
k_{1}\ar@{ |>->}[dr]^-{k_1}\ar@{-->}[dd]_{k_1} \\
& s \\
k_{2}\ar@{ |>->}[ur]_-{k_2}
}$$

\noindent showing that $k_{1} \leq k_{2}$ in $\KSub(s)$. Conversely,
if there is an $f\colon k_{1}\rightarrow k_{2}$ with $k_{2} \cdot f = k_{1}$,
then $k_{2}\cdot k_{1} = k_{2}\cdot k_{2} \cdot f = k_{2} \cdot f = k_{1}$,
showing that $k_{1}\leq k_{2}$ in $K_{s}$. \QED
\end{proof}

\subsection{Generators}\label{GeneratorSubsec}

Recall that a generator in a category is an object $I$ such that for
each pair of maps $f,g\colon X\rightarrow Y$, if $f\after x = g\after
x$ for all $x\colon I\rightarrow X$, then $f=g$. Every singleton set
is a generator in \Sets, and also in \Rel. The complex numbers
$\mathbb{C}$ form a generator in the category \Hilb of Hilbert spaces
of $\mathbb{C}$. And the two-element orthomodular lattice is a
generator in \Cat{OMLatGal} by Lemma~\ref{PointLem}.

We shall write $\Cat{DKCg} \hookrightarrow \Cat{DKC}$ for the subcategory
of dagger kernel categories with a given generator, and with morphisms
preserving the generator, up-to-isomorphism.

\begin{lemma}
The dagger kernel category $\dagKaroubi{S}$ associated with a Foulis
semigroup has the unit $1\in S$ as generator. The functor $\Cat{Fsg}
\rightarrow \Cat{DKC}$ from Theorem~\ref{Foulis2DagKerThm} restricts
to $\Cat{Fsg} \rightarrow \Cat{DKCg}$.
\end{lemma}

\begin{proof}
Assume $f,g\colon s\rightarrow t$ in $\dagKaroubi{S}$ with $f\after x
= g\after x$ for each map $x\colon 1\rightarrow s$. Then, in
particular for $x=s$ we get $f = f\after s = g\after s = g$. Every
morphism $h\colon S\rightarrow T$ of Foulis semigroups satisfies
$h(1) = 1$, so that the induced functor $\dagKaroubi{S} \rightarrow
\dagKaroubi{T}$ preserves the generator. \QED
\end{proof}

\begin{lemma}
The mapping $\cat{D} \mapsto \KSub(I)$ yields a functor
$\Cat{DCKg} \rightarrow \Cat{OMLat}$.
\end{lemma}

\begin{proof}
If $F\colon \cat{D} \rightarrow \cat{E}$ is a functor in $\Cat{DCKg}$,
then one obtains a mapping $\KSub_{\Cat{D}}(I) \rightarrow \KSub_{\Cat{E}}(I)$
by:
$$\xymatrix{
\Big(M\ar@{ |>->}[r]^-{m} & I\Big)\ar@{|->}[r] &
   \Big(FM\ar@{ |>->}[r]^-{Fm} & FI\ar[r]^-{\cong} & I\Big).
}$$

\noindent Since all the orthomodular structure in kernel posets
$\KSub(X)$ is defined in terms of kernels and daggers, it is preserved
by $F$. \QED
\end{proof}

By composition we obtain the original (``old'') way to construct an
orthomodular lattice out of a Foulis semigroup, see~\cite{Foulis63}.

\begin{corollary}
\label{Foulis2OMLatCor}
The composite functor $\Cat{Fsg} \rightarrow \Cat{DCKg} \rightarrow
\Cat{OMLat}$ maps a Foulis semigroup $S$ to the orthomodular lattice
$\sai{S} = \set{\sai{t}}{t\in S} = K_{1} \cong \KSub(1)$ from
Lemma~\ref{FoulisOMKerLem}, over the generator $1$. \QED
\end{corollary}

\auxproof{
We follow the proof in~\cite[Chapter~5, \S\S18]{Kalmbach83} in a
number of small steps, involving elements $s,t\in S$, often
using requirement $(4')$ from Definition~\ref{FoulisDef}.
\begin{enumerate}
\renewcommand{\theenumi}{\alph{enumi}}
\item $s\cdot t = 0$ iff $\sai{s}\cdot t = t$. For the (if)-part, assume
  $\sai{s}\cdot t = t$; then $s\cdot t = s\cdot \sai{s}\cdot t = 0\cdot t =
  0$.  Conversely, if $s\cdot t = 0$ then $t = \sai{\sai{t^{\dag} \cdot
  s^{\dag}}\cdot s}\cdot t = \sai{\sai{0^{\dag}}\cdot s}\cdot t = \sai{1\cdot
  s}\cdot t = \sai{s}\cdot t$.

\item $s = s\cdot \sai{\sai{s}}$, and thus $s\leq \sai{\sai{s}} =
  s^{\perp\perp}$. This is obtained by substituting $\sai{s}$ for $s$ and
  $s^{\dag}$ for $t$ in the equation $t = \sai{\sai{t^{\dag} \cdot
  s^{\dag}}\cdot s}\cdot t$ from Defintion~\ref{FoulisDef}~$(4')$. It
  yields:
$$\begin{array}{rcll}
s^{\dag}
& = &
\sai{\sai{s^{\dag\dag} \cdot \sai{s}^{\dag}}\cdot \sai{s}}\cdot s^{\dag} \\
& = &
\sai{\sai{s\cdot \sai{s}}\cdot \sai{s}}\cdot s^{\dag} \\
& = &
\sai{\sai{0}\cdot \sai{s}}\cdot s^{\dag} 
   & \mbox{by~$(4')$ in Definition~\ref{FoulisDef}} \\
& = &
\sai{\sai{s}}\cdot s^{\dag}
   & \mbox{since $\sai{0}=1$.}
\end{array}$$

\noindent Hence the result follows by applying the involution
$\dag$ on both sides.

\item $\sai{\sai{\sai{s}}} = \sai{s}$, so that $t^{\perp\perp} = t$ for 
$t\in \sai{S}$. To start, 
$$\begin{array}{rcll}
s\cdot \sai{\sai{\sai{s}}}
& = &
s\cdot \sai{\sai{s}} \cdot \sai{\sai{\sai{s}}} & \mbox{by~(b)} \\
& = &
s \cdot 0 & \mbox{by~$(4')$ in Definition~\ref{FoulisDef}} \\
& = &
0.
\end{array}$$

\noindent Hence we can apply~(a), so that $\sai{\sai{\sai{s}}} =
\sai{s} \cdot \sai{\sai{\sai{s}}} = \sai{s}$, where the last equation
follows by applying~(b) to $\sai{s}$.

\item $s\leq t$ implies $\sai{t} \leq \sai{s}$ (and thus $t^{\perp} \leq
  s^{\perp}$). We assume $s\cdot t = s$ and have to prove $\sai{t}\cdot
  \sai{s} = \sai{t}$. Since $\sai{t} = \sai{t}^{\dag}$ the following suffices.
$$\begin{array}{rcll}
\sai{t}
& = &
\sai{\sai{\sai{t}^{\dag} \cdot s^{\dag}}\cdot s}\cdot \sai{t}
   & \mbox{by~$(4')$ in Definition~\ref{FoulisDef}} \\
& = &
\sai{\sai{(s\cdot \sai{t})^{\dag}}\cdot s}\cdot \sai{t} \\
& = &
\sai{\sai{(s\cdot t \cdot \sai{t})^{\dag}}\cdot s}\cdot \sai{t} 
   & \mbox{by assumption} \\
& = &
\sai{\sai{(s\cdot 0)^{\dag}}\cdot s}\cdot \sai{t} 
   & \mbox{by~$(4')$ in Definition~\ref{FoulisDef}} \\
& = &
\sai{1\cdot s}\cdot \sai{t} 
   & \mbox{since $\sai{0}=1$} \\
& = &
\sai{s}\cdot \sai{t} \\
& = &
\sai{s}^{\dag}\cdot \sai{t}^{\dag} \\
& = &
(\sai{s}\cdot \sai{t})^{\dag}.
\end{array}$$

\item $s\leq \sai{t} \Rightarrow t\leq \sai{s}$, by~(d), (b) and
  transitivity of $\leq$. Hence $\sai{-}\colon S\op \rightarrow S$ is
  self-adjoint wrt.\ $\leq$, with $\sai{S} =
  \setin{s}{S}{s=\sai{\sai{s}}}$ as set of closed elements, with a
  Galois connection:
\begin{equation}
\label{PsquareGaloisEqn}
\xymatrix{
S\ar@<1ex>[rr]^-{\sai{\sai{-}} = (-)^{\perp}} & & \;\sai{S}\ar@{_{(}->}[ll]<1ex>
}
\end{equation}

\noindent since for $s\in S$ and $t\in \sai{S}$ one has $s\leq t$
iff $\sai{\sai{s}} \leq t$.

\item Multiplication $\cdot$ is idempotent on $\sai{S}$, so that $\leq$
  is reflexive on $\sai{S}$.  This is easy since $\sai{s} = \sai{s}\cdot
  \sai{\sai{\sai{s}}} = \sai{s}\cdot \sai{s}$ by~(b) and~(c).

\item The order $\leq$ is also symmetric on $\sai{S}$, since if $\sai{s} =
  \sai{s}\cdot \sai{t}$ and $\sai{t} = \sai{t}\cdot \sai{s}$, then:
$$\sai{s}
=
\sai{s}^{\dag}
=
(\sai{s}\cdot \sai{t})^{\dag}
=
\sai{t}^{\dag} \cdot \sai{s}^{\dag}
=
\sai{t}\cdot \sai{s}
=
\sai{t}.$$

\item Elements $t_{1},t_{2}\in \sai{S}$ have a meet $t_{1}\conjun t_{2} =
  \sai{\sai{r}}\in \sai{S}$, where $r = \sai{\sai{t_1}\cdot t_2}\cdot t_{2}$.
\begin{itemize}
\item $\sai{\sai{r}} \leq t_{1}$ holds since $\sai{t_{1}}\cdot t_{2} \cdot
  \sai{\sai{t_1}\cdot t_2} = 0$ by Definition~\ref{FoulisDef}~$(4')$. Hence
  by~(a) we get:
$$t_{2}\cdot \sai{\sai{t_1}\cdot t_2}
=
\sai{\sai{t_{1}}}\cdot t_{2} \cdot \sai{\sai{t_1}\cdot t_2}
=
t_{1}\cdot t_{2} \cdot \sai{\sai{t_1}\cdot t_2}.$$

\noindent But then $r\leq t_{1}$ follows:
$$\begin{array}{rcl}
r
\hspace*{\arraycolsep} = \hspace*{\arraycolsep}
r^{\dag\dag}
\hspace*{\arraycolsep} = \hspace*{\arraycolsep}
\big(\sai{\sai{t_1}\cdot t_2}\cdot t_{2}\big)^{\dag\dag} 
& = &
\big(t_{2}^{\dag} \cdot \sai{\sai{t_1}\cdot t_2}^{\dag}\big)^{\dag} \\
& = &
\big(t_{2} \cdot \sai{\sai{t_1}\cdot t_2}\big)^{\dag} \\
& = &
\big(t_{1}\cdot t_{2} \cdot \sai{\sai{t_1}\cdot t_2}\big)^{\dag} \\
& = &
\sai{\sai{t_1}\cdot t_2}^{\dag} \cdot t_{2}^{\dag} \cdot t_{1}^{\dag} \\
& = &
\sai{\sai{t_1}\cdot t_2} \cdot t_{2} \cdot t_{1} \\
& = &
r\cdot t_{1}.
\end{array}$$

\noindent Hence $\sai{\sai{r}} \leq \sai{\sai{t_{1}}} = t_{1}$.

\item Also $\sai{\sai{r}} \leq t_{2}$ since:
$$r\cdot \sai{t_{2}}
=
\sai{\sai{t_1}\cdot t_2} \cdot t_{2} \cdot \sai{t_{2}}
=
\sai{\sai{t_1}\cdot t_2} \cdot 0
=
0,$$

\noindent so that $\sai{r}\cdot \sai{t_{2}} = \sai{t_{2}}$ by~(a). Hence:
$$\sai{t_{2}}
=
\sai{t_{2}}^{\dag}
=
(\sai{r}\cdot \sai{t_{2}})^{\dag}
=
\sai{t_{2}}^{\dag}\cdot \sai{r}^{\dag}
=
\sai{t_{2}} \cdot \sai{r}.$$

\noindent This means $\sai{t_{2}} \leq \sai{r}$ and thus $\sai{\sai{r}} \leq
\sai{\sai{t_2}} = t_{2}$.

\item If also $s\in \sai{S}$ satisfies $s \leq t_{1}$ and $s\leq t_{2}$,
\textit{i.e.}~$s\cdot t_{1} = s = s\cdot t_{2}$, then, by
Definition~\ref{FoulisDef}~$(4')$,
$$\begin{array}{rcl}
\sai{t_1}\cdot t_{2}\cdot s
\hspace*{\arraycolsep} = \hspace*{\arraycolsep}
\big(s^{\dag} \cdot t_{2}^{\dag} \cdot \sai{t_1}^{\dag}\big)^{\dag}
& = &
\big(s \cdot t_{2} \cdot \sai{t_1}\big)^{\dag} \\
& = &
\big(s \cdot t_{1} \cdot \sai{t_1}\big)^{\dag} \\
& = &
(s \cdot 0)^{\dag} \\
& = &
0.
\end{array}$$

\noindent Hence $\sai{\sai{t_1}\cdot t_2}\cdot s = s$ by~(a). Then:
$$\begin{array}{rcl}
s\cdot r
& = &
s\cdot \sai{\sai{t_1}\cdot t_2}\cdot t_{2} \\
& = &
s^{\dag}\cdot \sai{\sai{t_1}\cdot t_2}^{\dag}\cdot t_{2} \\
& = &
\big(\sai{\sai{t_1}\cdot t_2}\cdot s\big)^{\dag} \cdot t_{2} \\
& = &
s^{\dag}\cdot t_{2} \\
& = &
s\cdot t_{2} \\
& = &
s.
\end{array}$$

\noindent Hence $s\leq r \leq \sai{\sai{r}}$.
\end{itemize}

\item $t\conjun t^{\perp} = 0$, for $t\in \sai{S}$, by a straightforward
calculation:
$$\begin{array}{rcll}
t \conjun t^{\perp}
& = &
\sai{\sai{\sai{\sai{t}\cdot \sai{t}}\cdot \sai{t}}} 
   & \mbox{by definition of $\conjun$ and $(-)^{\perp}$} \\
& = &
\sai{\sai{\sai{\sai{t}}\cdot \sai{t}}} & \mbox{by~(f)} \\
& = &
\sai{\sai{t\cdot \sai{t}}} & \mbox{by~(c)} \\
& = &
\sai{\sai{0}} & \mbox{by Definition~\ref{FoulisDef}~$(4')$} \\
& = &
0.
\end{array}$$

\item Finally, orthomodularity holds in $\sai{S}$. We assume $s\leq t$
  and $s^{\perp}\conjun t=0$ and have to show $t\leq s$ (and thus
  $s=t$). One gets:
$$\begin{array}{rcl}
0
\hspace*{\arraycolsep} = \hspace*{\arraycolsep}
s^{\perp}\conjun t
& = &
\sai{\sai{\sai{\sai{\sai{s}}\cdot t}\cdot t}} \\
& = &
\sai{\sai{\sai{s\cdot t}\cdot t}} \\
& = &
\sai{\sai{\sai{s}\cdot t}} \\
& \geq &
\sai{s}\cdot t.
\end{array}$$

\noindent Hence $\sai{\sai{s}}\cdot t = t$, by~(a), and thus $t = t^{\dag}
= (\sai{\sai{s}}\cdot t)^{\dag} = (s\cdot t)^{\dag} = t^{\dag} \cdot
s^{\dag} = t\cdot s$, so that $t\leq s$. \QED
\end{enumerate}
}

In the reverse direction we have seen in
Corollary~\ref{OMLat2FoulisCor} that the set $\EndoHom{X}$ of (Galois)
endomaps on an orthomodulair lattice $X$ is a Foulis semigroup, but
functoriality is problematic. However, we can now solve a problem that
was left open in~\cite{HeunenJ09a}, namely the construction of a
dagger kernel category out of an orthomodular lattice
$X$. Theorem~\ref{Foulis2DagKerThm} says that the dagger Karoubi
envolope $\dagKaroubi{\EndoHom{X}}$ is a dagger kernel category. Its
objects are self-adjoint idempotents $s\colon X\rightarrow X$, and its
morphisms $f\colon (X,s)\rightarrow (X,t)$ are maps $f\colon
X\rightarrow X$ in \Cat{OMLatGal} with $t \after f = f = f\after s$.

\section{Conclusions}\label{ConclusionSec}

There is a relatively recent line of research applying categorical
methods in quantum theory, see for
instance~\cite{ButterfieldI98,AbramskyC04,Selinger07,DoeringI08,HeunenLS09,CoeckePV09}. This
paper fits into this line of work, with a focus on quantum logic
(following~\cite{HeunenJ09a}), and establishes a connection to early
work on quantum structures. It constructs new (dagger kernel)
categories of orthomodular lattices and of self-adjoint idempotents in
Foulis semigroups (also known as Baer *-semigroups). These categorical
constructions are shown to generalise translations between
orthomodular lattics and Foulis semigroups from the 1960s.  They
provide a framework for the systematic study of quantum (logical)
structures.

The current (categorical logic) framework may be used to address some
related research issues. We mention three of them.
\begin{itemize}
\item As shown, the category \Cat{OMLatGal} of orthomodular lattices
  and Galois connections has (dagger) kernels and biproducts
  $\oplus$. An open question is whether it also has tensors $\otimes$,
  to be used for the construction of (logics of) compound systems. The
  existence of such tensors is a subtle matter, given the restrictions
  described in~\cite{RandallF79}.

\item A dagger kernel category gives rise to not just one orthomodular
  lattice (or Foulis semigroup), but to a collection, indexed by the
  objects of the category, see for instance the presheaf description
  in Proposition~\ref{DagKer2FoulisProp}. The precise, possibly
  sheaf-theoretic (see~\cite{GravesS73}), nature of this indexing is
  not fully understood yet.

\item So-called effect algebras have been introduced as more recent
  generalisations of orthomodular lattices, see~\cite{DvurecenskijP00}
  for an overview. An open question is how such quantum structures
  relate to the present approach.
\end{itemize}

\subsubsection*{Acknowledgements}

Many thanks to Chris Heunen for discussions and joint
work~\cite{HeunenJ09a}.


\begin{thebibliography}{10}

\bibitem{AbramskyC04}
S.~Abramsky and B.~Coecke.
\newblock A categorical semantics of quantum protocols.
\newblock In {\em Logic in Computer Science}, pages 415--425. IEEE, Computer
  Science Press, 2004.

\bibitem{BaltagS06}
A.~Baltag and S.~Smets.
\newblock {LQP}: the dynamic logic of quantum information.
\newblock {\em Math. Struct. in Comp. Sci.}, 16:491--525, 2006.

\bibitem{BirkhoffN36}
G.~Birkhoff and J.~von Neumann.
\newblock The logic of quantum mechanics.
\newblock {\em Ann. Math.}, 37:823--843, 1936.

\bibitem{BlythJ72}
T.S. Blyth and M.F. Janowitz.
\newblock {\em Residuation Theory}.
\newblock Pergamum Press, 1972.

\bibitem{ButterfieldI98}
J.~Butterfield and C.J. Isham.
\newblock A topos perspective on the {K}ochen-{S}pecker theorem: {I}. quantum
  states as generalized valuations.
\newblock {\em Int. Journ. Theor. Physics}, 37(11):2669–--2733, 1998.

\bibitem{CoeckePV09}
B.~Coecke, D.~Pavlovi{\'c}, and J.~Vicary.
\newblock A new description of orthogonal bases.
\newblock {\em Math. Struct. in Comp. Sci.}, 2009, to appear.
\newblock Available from \url{http://arxiv.org/abs/0810.0812}.

\bibitem{CoeckeS04}
B.~Coecke and S.~Smets.
\newblock The {Sasaki} hook is not a [static] implicative connective but
  induces a backward [in time] dynamic one that assigns causes.
\newblock {\em Int. Journ. of Theor. Physics}, 43(7/8):1705--1736, 2004.

\bibitem{DoeringI08}
A.~D{\"o}ring and C.J. Isham.
\newblock A topos foundation for theories of physics: {I - IV}.
\newblock {\em Journ. of Math. Physics}, 49:053515--053518, 2008.

\bibitem{Dvurecenskij92}
A.~Dvure\v{c}enskij.
\newblock {\em Gleason's Theorem and Its Applications}.
\newblock Number~60 in Mathematics and its Applications. Kluwer Acad. Publ.,
  Dordrecht, 1992.

\bibitem{DvurecenskijP00}
A.~Dvure\v{c}enskij and S.~Pulmannov{\'a}.
\newblock {\em New Trends in Quantum Structures}.
\newblock Kluwer Acad. Publ., Dordrecht, 2000.

\bibitem{Finch70}
P.~D. Finch.
\newblock Quantum logic as an implication algebra.
\newblock {\em Bull. Amer. Math. Soc.}, 2:101--106, 1970.

\bibitem{Foulis60}
D.~J. Foulis.
\newblock {Baer} *-semigroups.
\newblock {\em Proc. Amer. Math. Soc.}, 11:648--654, 1960.

\bibitem{Foulis62}
D.~J. Foulis.
\newblock A note on orthomodular lattices.
\newblock {\em Portugaliae Mathematica}, 21:65--72., 1962.

\bibitem{Foulis63}
D.~J. Foulis.
\newblock Relative inverses in {Baer} *-semigroups.
\newblock {\em Michigan Math. Journ.}, 10(1):65--84, 1963.

\bibitem{Freyd64}
P.J. Freyd.
\newblock {\em Abelian Categories: An Introduction to the Theory of Functors}.
\newblock Harper and Row, New York, 1964.
\newblock Available via \url{www.tac.mta.ca/tac/reprints/articles/3/tr3.pdf}.

\bibitem{GravesS73}
W.H. Graves and S.A. Selesnick.
\newblock An extension of the {Stone} representation for orthomodular lattices.
\newblock {\em Coll. Mathematicum}, XXVII:21--30, 1973.

\bibitem{Hayashi85}
S.~Hayashi.
\newblock Adjunction of semifunctors: categorical structures in nonextensional
  lambda calculus.
\newblock {\em Theor. Comp. Sci.}, 41:95--104, 1985.

\bibitem{Heunen09}
C.~Heunen.
\newblock {\em Categorical Quantum Models and Logics}.
\newblock PhD thesis, Univ. Nijmegen, 2009.

\bibitem{HeunenJ09a}
C.~Heunen and B.~Jacobs.
\newblock Quantum logic in dagger kernel categories.
\newblock In B.~Coecke, P.~Panangaden, and P.~Selinger, editors, {\em
  Proceedings of the 6th International Workshop on Quantum Programming
  Languages (QPL 2009)}, Elect. Notes in Theor. Comp. Sci. Elsevier, Amsterdam,
  2009, to appear.
\newblock Available from \url{http://arxiv.org/abs/0902.2355}.

\bibitem{HeunenLS09}
C.~Heunen, N.P. Landsman, and B.~Spitters.
\newblock A topos for algebraic quantum theory.
\newblock {\em Comm. in Math. Physics}, 2009, to appear.
\newblock Available from \url{http://arxiv.org/abs/0709.4364}.

\bibitem{Hoofman92}
R.~Hoofman.
\newblock {\em Non-Stable Models of Linear Logic}.
\newblock PhD thesis, Univ. Utrecht, 1992.

\bibitem{HoofmanM95}
R.~Hoofman and I.~Moerdijk.
\newblock A remark on the theory of semi-functors.
\newblock {\em Math. Struct. in Comp. Sci.}, 5(1):1--8, 1995.

\bibitem{Husimi37}
K.~Husimi.
\newblock Studies on the foundation of quantum mechanics {I}.
\newblock {\em Proc. physicomath. Soc. Japan}, 19:766--789, 1937.

\bibitem{Jacobs91b}
B.~Jacobs.
\newblock Semantics of the second order lambda calculus.
\newblock {\em Math. Struct. in Comp. Sci.}, 1(3):327--360, 1991.

\bibitem{Jacobs99a}
B.~Jacobs.
\newblock {\em Categorical Logic and Type Theory}.
\newblock North Holland, Amsterdam, 1999.

\bibitem{Kalmbach83}
G.~Kalmbach.
\newblock {\em Orthomodular Lattices}.
\newblock Academic Press, London, 1983.

\bibitem{Karoubi78}
M.~Karoubi.
\newblock {\em K-theory. An Introduction}.
\newblock Springer, 1978.

\bibitem{LambekS86}
J.~Lambek and P.J. Scott.
\newblock {\em Introduction to higher order Categorical Logic}.
\newblock Number~7 in Cambridge Studies in Advanced Mathematics. Cambridge
  Univ. Press, 1986.

\bibitem{MacLaneM92}
S.~Mac Lane and I.~Moerdijk.
\newblock {\em Sheaves in Geometry and Logic. A First Introduction to Topos
  Theory}.
\newblock Springer, New York, 1992.

\bibitem{RandallF79}
C.~Randall and D.J. Foulis.
\newblock Tensor products of quantum logics do not exist.
\newblock {\em Notices Amer. Math. Soc.}, 26(6):A--557, 1979.

\bibitem{Scott80a}
D.S. Scott.
\newblock Relating theories of the $\lambda$-calculus.
\newblock In J.R. Hindley and J.P. Seldin, editors, {\em To H.B. Curry: Essays
  on Combinatory Logic, Lambda Calculus and Formalism}, pages 403--450, New
  York and London, 1980. Academic Press.

\bibitem{Selinger07}
P.~Selinger.
\newblock Dagger compact closed categories and completely positive maps
  (extended abstract).
\newblock In P.~Selinger, editor, {\em Proceedings of the 3rd International
  Workshop on Quantum Programming Languages (QPL 2005)}, number 170 in Elect.
  Notes in Theor. Comp. Sci., pages 139--163. Elsevier, Amsterdam, 2007.
\newblock DOI \url{http://dx.doi.org/10.1016/j.entcs.2006.12.018}.

\bibitem{Selinger08}
P.~Selinger.
\newblock Idempotents in dagger categories (extended abstract).
\newblock In P.~Selinger, editor, {\em Proceedings of the 4th International
  Workshop on Quantum Programming Languages (QPL 2006)}, number 210 in Elect.
  Notes in Theor. Comp. Sci., pages 107--122. Elsevier, Amsterdam, 2008.
\newblock DOI \url{http://dx.doi.org/10.1016/j.entcs.2008.04.021}.

\end{thebibliography}

\end{document}